\documentclass[11pt]{article}

\usepackage[english]{babel}
\usepackage{graphicx,color}
\usepackage{endnotes}
\usepackage{amsmath,amssymb,amsfonts,float,dsfont}
\usepackage {algorithm}
\usepackage{xspace}
\renewcommand{\And}{\wedge}
\newcommand{\Or}{\vee}

\newcommand{\variable}[1]{#1}

\newcommand{\st}{\variable{st}}
\newcommand{\status}{\variable{st}}

\newcommand{\dist}{\variable{d}}
\newcommand{\reset}{\variable{reset}}
\newcommand{\beRoot}{\variable{beRoot}}

\newcommand{\ComputeVariable}{\variable{cmpVar}}

\newcommand{\BestPointer}{\variable{bestPtr}}
\newcommand{\compute}{\variable{compute}}
\newcommand{\update}{\variable{upd}}
\newcommand{\VColor}{\variable{col}}
\newcommand{\VSatisfaction}{\variable{scr}}
\newcommand{\VCanQuit}{\variable{canQ}}
\newcommand{\VPointer}{\variable{ptr}}

\newcommand{\satisfaction}{realScr\xspace}
\newcommand{\numberYes}{\#InAll\xspace}

\newcommand{\N}{\variable{N}}

\newcommand{\A}{{\tt SDR}\xspace}
\newcommand{\DA}{{\tt FGA}\xspace}
\newcommand{\DU}{{\tt U}\xspace}

\DeclareMathOperator*{\argmin}{\arg\!\min}

\newcommand{\PB}{\mathbf{P\_RB}\xspace}
\newcommand{\PF}{\mathbf{P\_RF}\xspace}

\newcommand{\PC}{\mathbf{P\_C}\xspace}
\newcommand{\PRa}{\mathbf{P\_R1}\xspace}
\newcommand{\PRb}{\mathbf{P\_R2}\xspace}

\newcommand{\PRUp}{\mathbf{P\_Up}\xspace}
\newcommand{\PRoot}{\mathbf{P\_root}\xspace}
\newcommand{\PClean}{\mathbf{P\_Clean}\xspace}

\newcommand{\PQuitAlliance}{\mathbf{P\_canQuit}\xspace}
\newcommand{\PBecomeNormal}{\mathbf{P\_toQuit}\xspace}
\newcommand{\PReset}{\mathbf{P\_reset}\xspace}

\newcommand{\PICorrect}{\mathbf{P\_ICorrect}\xspace}
\newcommand{\PRulePointer}{\mathbf{P\_updPtr}\xspace}
\newcommand{\correct}{\mathbf{P\_Correct}\xspace}

\newcommand{\ruleRB}{\mathbf{rule\_RB}\xspace}
\newcommand{\ruleRF}{\mathbf{rule\_RF}\xspace}
\newcommand{\ruleC}{\mathbf{rule\_C}\xspace}
\newcommand{\ruleR}{\mathbf{rule\_R}\xspace}
\newcommand{\wordsI}{\mathbf{words\_I}\xspace}

\newcommand{\ruleColor}{\mathbf{rule\_Clr}\xspace}
\newcommand{\rulePointerA}{\mathbf{rule\_P1}\xspace}
\newcommand{\rulePointerB}{\mathbf{rule\_P2}\xspace}
\newcommand{\ruleCanQuitB}{\mathbf{rule\_Q}\xspace}

\newcommand{\ParU}{\variable{K}}
\newcommand{\val}{\variable{c}}
\newcommand{\PAgree}{\mathbf{P\_Ok}\xspace}

\newcommand{\PUp}{\mathbf{P\_Up}\xspace}

\newcommand{\ruleUA}{\mathbf{rule\_U}\xspace}

\newenvironment{proof}{\noindent {\em Proof. } }{{\hfill\
    $\Box$}\vspace{.5pc}} 

\newtheorem{lemma}{Lemma}
\newtheorem{corollary}{Corollary}
\newtheorem{theorem}{Theorem}
\newtheorem{remark}{Remark}
\newtheorem{definition}{Definition}
\newtheorem{property}{Property}

\renewcommand{\bmod}{\text{\%}}

\newcommand{\prewordsI}{\mathbf{prewords\_I}\xspace}
\newcommand{\confI}{\mathbf{cw\_I}\xspace}

\usepackage{fullpage}
\usepackage{times}

\title{Self-Stabilizing Distributed Cooperative Reset\thanks{This
    study was partially supported by the French \textsc{anr} projects
    \textsc{ANR-16-CE40-0023} (\textsc{descartes}) and \textsc{ANR-16
      CE25-0009-03} (\textsc{estate}).}}

\author{ St\'ephane Devismes$^\dag$ and Colette Johnen$^\ddag$\\
  \\
  {\small $^\dag$ Universit\'e Grenoble Alpes, VERIMAG, UMR 5104, France}\\
  {\small $^\ddag$ Universit\'e de Bordeaux, LaBRI, UMR 5800, France}
}

\date{}

\begin{document}

\maketitle

\begin{abstract}
Self-stabilization is a versatile fault-tolerance approach that
characterizes the ability of a system to eventually resume a correct
behavior after any finite number of transient faults.
In this paper, we propose a {\em self-stabilizing reset algorithm}
working in {\em anonymous} networks. This algorithm resets the network
in a {\em distributed} non-centralized manner, {\em i.e.}, it is
multi-initiator, as each process detecting an inconsistency may
initiate a reset. It is also {\em cooperative} in the sense that it
coordinates concurrent reset executions in order to gain efficiency.
Our approach is general since our reset algorithm allows to build
self-stabilizing solutions for various problems and settings. As a
matter of fact, we show that it applies to both static and dynamic
specifications since we propose efficient self-stabilizing reset-based
algorithms for the (1-minimal) $(f,g)$-alliance (a generalization of
the dominating set problem) in identified networks and the unison
problem in anonymous networks. Notice that these two latter
instantiations enhance the state of the art. Indeed, in the former
case, our solution is more general than the previous ones; while in
the latter case, the time complexity of the proposed unison algorithm is
better than that of previous solutions of the literature.

\end{abstract}


\section{Introduction}

In distributed systems, a {\em self-stabilizing} algorithm is able to
recover a correct behavior in finite time, regardless of the
\emph{arbitrary} initial configuration of the system, and therefore
also after a finite number of transient faults, provided that those faults
do not alter the code of the processes.


For more than 40 years, a vast literature on self-stabilizing
algorithms has been developed. Self-stabilizing solutions have been
proposed for many kinds of classical distributed problems, {\em e.g.},
token circulation~\cite{HuangC93}, spanning tree
construction~\cite{CYH91}, clustering~\cite{CaronDDL10},
routing~\cite{Dolev1997122}, propagation of information with
feedback~\cite{BuiDPV99}, clock synchronization~\cite{CouvreurFG92},
{\em etc}. Moreover, self-stabilizing algorithms have been designed to
handle various environments, {\em e.g.}, wired
networks~\cite{HuangC93,CYH91,CaronDDL10,Dolev1997122,BuiDPV99,CouvreurFG92},
WSNs \cite{BenOthman2013199,Telematik_SSS_2013_Neighborhood},
peer-to-peer
systems~\cite{DBLP:journals/ppl/CaronDPT10,Caron20131533}, {\em etc}.
Drawing on this experience, general methodologies for making
distributed algorithms self-stabilizing have been proposed. In
particular, Katz and Perry~\cite{KP93j} give a characterization of
problems admitting a self-stabilizing solution. Precisely, they
describe a general algorithm that transforms almost all algorithms
(specifically, those algorithms that can be self-stabilized) into their
corresponding stabilizing version. However, this so-called {\em
  transformer} is, by essence, inefficient both in terms of space and
time complexities: actually, its purpose is only to demonstrate the feasibility
of the transformation.

Interestingly, many proposed general
methods~\cite{KP93j,APVD94c,AG94j,AO94c} are based on {\em reset}
algorithms. Such algorithms are initiated when an inconsistency is
discovered in the network, and aim at reinitializing the system to a
correct (pre-defined) configuration.

A reset algorithm may be centralized at a leader process ({\em e.g.},
see~\cite{AG94j}), or fully distributed, meaning multi-initiator (as
our proposal here).  In the former case, either the reset is coupled
with a snapshot algorithm (which makes a {\em global checking} of the
network), or processes detecting an incoherence (using {\em local
  checking}~\cite{APV91c}) should request a reset to the leader.  In
the fully distributed case, resets are locally initiated by processes
detecting inconsistencies. This latter approach is considered as more
efficient when the concurrent resets are coordinated. In other words,
concurrent resets have to be {\em cooperative} (in the sense
of~\cite{YNKM18c}) to ensure the fast convergence of the system to a
consistent global state.

Self-stabilization makes no hypotheses on the nature ({\em e.g.},
memory corruption or topological changes) or extent of transient
faults that could hit the system, and a self-stabilizing system
recovers from the effects of those faults in a unified manner.
Now, such versatility comes at a price, {\em e.g.}, after transient
faults cease, there is a finite period of time, called the {\em
  stabilization phase}, during which the safety properties of the
system are violated. Hence, self-stabilizing algorithms are mainly
compared according to their {\em stabilization time}, the maximum
duration of the stabilization phase.
 
General schemes and efficiency are usually understood as orthogonal
issues. We tackle this problem by proposing an efficient
self-stabilizing reset algorithm working in any anonymous connected
network. Our algorithm is written in the locally shared memory model
with composite atomicity, where executions proceed in atomic steps (in
which a subset of enabled processes move, {\em i.e.}, update their
local states) and the asynchrony is captured by the notion of {\em
  daemon}. The most general daemon is the {\em distributed unfair
  daemon}. So, solutions stabilizing under such an assumption are
highly desirable, because they work under any other daemon assumption.

The {\em stabilization time} is usually evaluated in terms of
rounds, which capture the execution time according to the speed of the
slowest processes. But, another crucial issue is the number of local
state updates, {\em i.e.} the number of {\em moves}. Indeed, the
stabilization time in moves captures the amount of computations an
algorithm needs to recover a correct behavior. 

The daemon assumption and time complexity are closely related.  To
obtain practical solutions, the designer usually tries to avoid strong
assumptions on the daemon, like for example, assuming all executions
are synchronous.  Now, when the considered daemon does not enforce any
bound on the execution time of processes, the stabilization time
in moves can be bounded only if the algorithm works under an unfair
daemon. For example, if the daemon is assumed to be {\em distributed
  and weakly fair} (a daemon stronger than the distributed unfair one)
and the studied algorithm actually requires the weakly fairness
assumption to stabilize, then it is possible to construct executions
whose convergence is arbitrarily long in terms of atomic steps (and so
in moves), meaning that, in such executions, there are processes whose
moves do not make the system progress in the convergence. In other
words, these latter processes waste computation power and so energy.
Such a situation should be therefore
prevented, making the unfair daemon more desirable than the weakly
fair one.

There are many self-stabilizing algorithms proven under the
distributed unfair daemon, {\em
  e.g.},~\cite{ACDDP14,DLP11,GHIJ14}.  However, analyzes
of the stabilization time in moves is rather unusual and this may be
an important issue.
Indeed, recently, several self-stabilizing algorithms which work under
a distributed unfair daemon have been shown to have an exponential
stabilization time in moves in the worst case, {\em e.g.}, the silent
leader election algorithms from~\cite{DLP11} (see~\cite{ACDDP14}) and the Breadth-First Search (BFS) algorithm of Huang
and Chen~\cite{HC92} (see~\cite{DJ16}).

\subsection{Contribution}

We propose an efficient self-stabilizing reset algorithm working in
any anonymous connected network. Our algorithm is written in the
locally shared memory model with composite atomicity, assuming a
distributed unfair daemon, {\em i.e.}, the most general scheduling
assumption of the model. 
It is based on local checking and is fully distributed ({\em i.e.},
multi-initiator). Concurrent resets are locally initiated by processes
detecting inconsistencies, these latter being cooperative to gain
efficiency.

As a matter of fact, our algorithm makes an input algorithm
recovering a consistent global state within at most $3n$ rounds, where
$n$ is the number of processes. During a recovering, any process executes at most $3n+3$ moves. 
Our reset algorithm allows to build efficient self-stabilizing
solutions for various problems and settings. In particular, it applies
to both static and dynamic specifications. In the static case, the
self-stabilizing solution we obtain is also {\em silent}~\cite{DolevGS96}: a
silent algorithm converges within finite time to a configuration from
which the values of the communication registers used by the algorithm
remain fixed.
Silence is a desirable property. Indeed, as noted 
in~\cite{DolevGS96}, the silent property usually implies more 
simplicity in the algorithm design. Moreover, a silent algorithm  
may utilize less communication operations and communication 
bandwidth.

To show the efficiency of our method, we propose two reset-based
self-stabilizing algorithms, respectively solving the unison problem in
anonymous networks and the 1-minimal
$(f,g)$-alliance in identified networks.

Our unison algorithm has a stabilization time in $O(n)$ rounds and
$O(\Delta.n^2)$ moves. Actually, its stabilization times in round
matches the one of the previous best existing solution~\cite{BPV04c}. However,
it achieves a better stabilization time in moves, since the algorithm
in~\cite{BPV04c} stabilizes in $O(D.n^3+\alpha.n^2)$ moves (as shown
in~\cite{DP12}), where $\alpha$ is greater than the length of
the longest chordless cycle in the network.

As explained before, our 1-minimal $(f,g)$-alliance algorithm is also
silent. Its stabilization time is $O(n)$ rounds and $O(\Delta.n.m)$
moves, where $D$ is the network diameter and $m$ is the number of
edges in the network. To the best of our knowledge, until now there
was no self-stabilizing algorithm solving that problem without any
restriction on $f$ and $g$.

\subsection{Related Work}

Several reset algorithms have been proposed in the literature.  In
particular, several solutions, {\em e.g.},~\cite{APVD94c,APV91c}, have
been proposed in the I/O automata model. In this model, communications
are implemented using message-passing and assuming {\em weakly
  fairness}. Hence, move complexity cannot be evaluated in that
model. In these papers, authors additionally assume links with known
bounded capacity. In \cite{APVD94c}, authors introduce the notion of
local checking, and propose a method that, given a self-stabilizing
global reset algorithm, builds a self-stabilizing solution of any {\em
  locally checkable} problem ({\em i.e.}, a problem where
inconsistency can be locally detected) in an {\em identified}
network. The stabilization time in rounds of obtained solutions
depends on the input reset algorithm. In \cite{APV91c}, authors focus
on an restrictive class of locally checkable problems, those that are
also {\em locally correctable}. A problem is locally correctable if
the global configuration of the network can be corrected by applying
independent corrections on pair neighboring processes. Now, for
example, the 1-minimal $(f,g)$ alliance problem is not locally
correctable since there are situations in which the correction of a
single inconsistency may provoke a global correction in a domino
effect reaction. Notice also that processes are not assumed to be
identified in \cite{APV91c}, however the considered networks are {\em not
fully anonymous} either. Indeed, each link has one of its incident
processes designated as leader. Notice also that authors show a
stabilization time in $O(H)$ when the network is a tree, where $H$ is
the tree height.

Self-stabilization by power supply \cite{AB98j} also assumes
message-passing with links of known bounded capacity and process
{\em identifiers}.  Using this technique, the stabilization time is in
$O(n)$ rounds in general.  Now, only {\em static} problems, {\em e.g.}
leader election and spanning tree construction, are considered.

Fully anonymous networks are considered in \cite{AO94c} in
message-passing systems with unit-capacity links and assuming {\em weakly
fairness}. The proposed self-stabilizing reset has a memory requirement
in $O(\log^\star(n))$ bits per process. But this small complexity comes
at the price of a stabilization time in {\em $O(n\log n)$ rounds}.

Finally, Arora and Gouda have proposed a mono-initiator reset
algorithm in the locally shared memory model with composite
atomicity. Their self-stabilizing reset works in {\em identified} networks,
assuming a distributed {\em weakly fair} daemon. The stabilization time of
their solution is in {\em $O(n+\Delta.D)$ rounds}, where $\Delta$ is the
degree of the network.

\subsection{Roadmap}
The remainder of the paper is organized as follows.  In the next
section, we present the computational model and basic definitions.
In Section~\ref{sec:algo}, we present, prove, and analyze the time
complexity of our reset algorithm. In the two last sections, we
propose two efficient instances of our reset-based method, respectively
solving the unison problem in anonymous networks and the 1-minimal
$(f,g)$-alliance in identified networks.

\section{Preliminaries}\label{model}

\subsection{Network}\label{sub:network}

We consider a distributed system made of $n$ interconnected processes.
Information exchanges are assumed to be
bidirectional. Henceforth, the communication network is conveniently
modeled by a simple undirected connected graph $G=(V,E)$, where $V$ is
the set of processes and $E$ a set of $m$ edges $\{u,v\}$ representing
the ability of processes $u$ and $v$ to directly exchange information
together. We denote by $D$ the diameter of $G$, {\em i.e.}, the
maximum distance between any two pairs of processes.  For every edge
$\{u,v\}$, $u$ and $v$ are said to be {\em neighbors}. For every
process $u$, we denote by $\delta_u$ the degree of $u$ in $G$, {\em
  i.e.}, the number of its neighbors. Let $\Delta = \max_{u \in V}
\delta_u$ be the (maximum) degree of $G$.

\subsection{Computational Model}

We use the {\em composite atomicity model of computation}~\cite{D74j}
in which the processes communicate using a finite number of locally
shared registers, simply called {\em variables}.  Each process can
read its own variables and that of its neighbors, but can write only
to its own variables.  The \textit{state} of a process
is defined by the values of its variables.  A {\em
  configuration} of the system is a vector consisting of the states of
each process.

Every process~$u$ can access the states of its neighbors using a {\em
  local labeling}.  Such labeling is called {\em indirect naming} in
the literature~\cite{Sloman:1987}.  All labels of $u$'s neighbors are
stored into the set $\N(u)$.  To simplify the design of our
algorithms, we sometime consider the {\em closed neighborhood} of a
process $u$, {\em i.e.}, the set including $u$ itself and all its
neighbors. Let $\N[u]$ be the set of labels local to $u$ designating
all members of its closed neighborhood, including $u$ itself. In
particular, $\N(u) \subsetneq \N[u]$. We assume that each process $u$
can identify its local label $\alpha_u(v)$ in the sets $\N(v)$ of each
neighbor $v$ and $\N[w]$ of each member $w$ of its closed
neighborhood.  When it is clear from the context, we use, by an abuse
of notation,~$u$ to designate both the process $u$ itself, and its
local labels ({\em i.e.}, we simply use $u$ instead of~$\alpha_u(v)$
for~$v \in \N[u]$).

A \emph{distributed algorithm\/} consists of one local 
program per process. 
The \textit{program} of each process consists of a finite
set of \textit{rules} of the form\ $$\langle label \rangle\ :\ \langle guard \rangle\ \to\ \langle action \rangle$$
{\em Labels} are only used to identify rules in the reasoning.  A
\textit{guard} is a Boolean predicate involving the state of the
process and that of its neighbors.  The {\em action\/} part 
of a rule updates the state of the process.  
A rule can be executed only if its
guard evaluates to {\em true}; in this case, the rule is 
said to be {\em enabled\/}.  
A process is said to be enabled if at least one of
its rules is enabled.  We denote by 
$\mbox{\it Enabled}(\gamma)$ the
subset of processes that are enabled in configuration~$\gamma$.

When the configuration is $\gamma$ and $\mbox{\it Enabled}(\gamma)
\neq \emptyset$, a non-empty set $\mathcal X \subseteq \mbox{\it
  Enabled}(\gamma)$ is {\em activated} by a so-called {\em daemon}; then every
process of $\mathcal X$ {\em atomically} executes one of its enabled
rules,\footnote{In case of several enabled actions at the activated
  process, the choice of the executed action is nondeterministic.}
leading to a new configuration $\gamma^\prime$, and so on.  The
transition from $\gamma$ to $\gamma^\prime$ is called a {\em step}.
The possible steps induce a binary relation over the set of
configurations, denoted by $\mapsto$.  An {\em execution\/} is a
maximal sequence of configurations $e=\gamma_0\gamma_1\cdots
\gamma_i\cdots$ such that $\gamma_{i-1}\mapsto\gamma_i$ for all $i>0$.
The term ``maximal'' means that the execution is either infinite, or
ends at a {\em terminal\/} configuration in which no rule is enabled
at any process.

Each step from a configuration to another is driven by a daemon.  We
define a daemon as a predicate $\mathfrak{D}$ over executions.  A daemon $\mathfrak{D}$ may
restrain the set of possible executions (in particular, it may forbid
some steps), {\em i.e.}, only executions satisfying $\mathfrak{D}$ are possible.
We assume here the daemon is {\em distributed} and {\em unfair}.
``Distributed'' means that while the configuration is not terminal,
the daemon should select at least one enabled process, maybe
more. ``Unfair'' means that there is no fairness constraint, {\em
  i.e.}, the daemon might never select an enabled process unless it is
the only enabled process. In other words, the distributed unfair
daemon is defined by the predicate {\em true} ({\em i.e.} it is the most
general daemon), and assuming that daemon, every execution is possible
and $\mapsto$ is actually the set of all possible steps.

\subsection{Self-Stabilization and Silence}

Let ${\tt A}$ be a distributed algorithm.   Let $P$
and $P'$ be two predicates over configurations of ${\tt A}$. Let $C$ and
$C'$ be two subsets of $\mathcal C_{{\tt A}}$, the set of ${\tt A}$'s
configurations.

\begin{itemize}
\item $P$ (resp. $C$) is {\em closed by} ${\tt A}$ if for every step
  $\gamma \mapsto \gamma'$ of ${\tt A}$, $P(\gamma) \Rightarrow
  P(\gamma')$ (resp. $\gamma \in C \Rightarrow \gamma' \in C$).
  
\item ${\tt A}$ {\em converges from} $P'$ (resp. $C'$) {\em to} $P$
  (resp. $C$)  if each of its executions
   starting from a configuration satisfying $P'$ (resp. in a
  configuration of $C'$) contains a configuration satisfying $P$
  (resp. a configuration of $C$).

\item $P$ (resp. $C$) is an {\em attractor} for ${\tt A}$ if $P$ (resp. $C$) is closed by ${\tt A}$
  and ${\tt A}$ converges from $true$ (resp. from $\mathcal C_{{\tt
      A}}$) to $P$ (resp. to $C$).
\end{itemize}
Let $SP$ be a specification, {\em i.e.}, a predicate over executions.
Algorithm ${\tt A}$ is {\em self-stabilizing} for $SP$ (under the
unfair daemon) if there exists a non-empty subset of its
configurations $\mathcal L$, called the {\em legitimate}
configurations, such that $\mathcal L$ is an attractor for ${\tt A}$
and every execution of ${\tt A}$ that starts in a configuration of
$\mathcal L$ satisfies $SP$. Configurations of $\mathcal C_{{\tt A}}
\setminus \mathcal L$ are called the {\em illegitimate}
configurations.

In our model, an algorithm is {\em silent}~\cite{DolevGS96} if and
only if all its possible executions are finite. Let $SP'$ be an
predicate over configurations of ${\tt A}$. Usually, silent
self-stabilization is (equivalently) reformulated as follows.  ${\tt
  A}$ is {\em silent and self-stabilizing} for the $SP'$ if all its
executions are finite and all its terminal configurations satisfy
$SP'$. Of course, in silent self-stabilization, the set of legitimate
configurations is chosen as the set of terminal configurations.

\subsection{Time Complexity}

We measure the time complexity of an algorithm using two notions: {\em
  rounds}~\cite{DIM93,CDPV02c} and {\em moves}~\cite{D74j}. 
 We say that a process {\em moves} in
$\gamma_i\mapsto\gamma_{i+1}$ when it executes a rule in
$\gamma_i\mapsto\gamma_{i+1}$.

The definition of round uses the concept of {\em
  neutralization}: a process~$v$ is \textit{neutralized} during a
step~$\gamma_i \mapsto \gamma_{i+1}$, if~$v$ is enabled in~$\gamma_i$
but not in configuration~$\gamma_{i+1}$, and it is not activated in the step~$\gamma_i \mapsto \gamma_{i+1}$.  

Then, the rounds are
inductively defined as follows. The first round of an execution~$e =
\gamma_0\gamma_1\cdots$ is the minimal prefix~$e' = \gamma_0\cdots\gamma_j$, such that every process that is enabled
in~$\gamma_0$ either executes a rule or is neutralized during a step
of~$e'$. Let~$e''$ be the suffix $\gamma_j\gamma_{j+1}\cdots$
of~$e$.  The second round of~$e$ is the first round of~$e''$, and so
on.

The {\em stabilization time} of a self-stabilizing algorithm is
the maximum time (in moves or rounds) over every possible execution (starting from any initial
configuration) to reach a legitimate configuration.

\subsection{Composition}

We denote by ${\tt A}$ $\circ$ ${\tt B}$ the composition of the two
algorithms ${\tt A}$ and ${\tt B}$ which is the distributed algorithm where
the local program (${\tt A}$ $\circ$ ${\tt B}$)($u$), for every process
$u$, consists of all variables and rules of both ${\tt A}$($u$) and
${\tt B}$($u$).


\section{Self-Stabilizing Distributed Reset Algorithm}\label{sec:algo}

\subsection{Overview of the Algorithm}

In this section, we present our distributed cooperative reset
algorithm, called \A. The formal code of \A, for each process $u$, is
given in Algorithm~\ref{alg:A}. This algorithm aims at reinitializing
an input algorithm {\tt I} when necessary. \A is self-stabilizing in
the sense that the composition {\tt I} $\circ$ \A is self-stabilizing
for the specification of {\tt I}. Algorithm \A works in anonymous
networks and is actually is multi-initiator: a process $u$ can
initiate a reset whenever it locally detects an inconsistency in {\tt
  I}, {\em i.e.}, whenever the predicate $\neg\PICorrect(u)$ holds
({\em i.e.}, {\tt I} is locally checkable). So, several resets may
be executed concurrently. In this case, they are coordinated: a reset
may be partial since we try to prevent resets from overlapping.

\subsection{The Variables}

Each process $u$ maintains two variables in Algorithm \A: $\status_u
\in \{C,RB,RC\}$, the {\em status} of $u$ with respect to the reset,
and $\dist_u \in \mathds{N}$, the {\em distance} of $u$ in a reset.

\paragraph{Variable $\status_u$.}
If $u$ is not currently involved into a reset, then it has status $C$,
which stands for {\em correct}. Otherwise, $u$ has status either $RB$
or $RF$, which respectively mean {\em reset broadcast} and {\em reset
  feedback}. Indeed, a reset is based on a (maybe partial) {\em
  Propagation of Information with Feedback (PIF)} where
processes reset their local state in {\tt I} (using the macro
$\reset$) during the broadcast phase.  When a reset locally terminates
at process $u$ ({\em i.e.}, when $u$ goes back to status $C$ by
executing $\ruleC(u)$), each member $v$ of its closed neighborhood
satisfies $\PReset(v)$, meaning that they are in a pre-defined initial
state of {\tt I}. At the global termination of a reset, every process
$u$ involved into that reset has a state in {\tt I} which is
consistent {\em w.r.t.} that of its neighbors, {\em i.e.},
$\PICorrect(u)$ holds.  Notice that, to ensure that $\PICorrect(u)$
holds at the end of a reset and for liveness issues, we enforce each
process $u$ stops executing {\tt I} whenever a member of its closed
neighborhood (in particular, the process itself) is involved into a reset:
whenever $\neg \PClean(u)$ holds, $u$ is not allowed to execute {\tt
  I}.

\paragraph{Variable $\dist_u $.}
This variable is meaningless when $u$ is not involved into a reset
({\em i.e.}, when $u$ has status $C$). Otherwise, the distance values
are used to arrange processes involved into resets as a {\em Directed
  Acyclic Graph (DAG)}. This distributed structure allows to prevent
both livelock and deadlock.  Any process $u$ initiating a reset (using
rule $\ruleR(u)$) takes distance 0. Otherwise, when a reset is
propagated to $u$ ({\em i.e.}, when $\ruleRB(u)$ is executed),
$\dist_u$ is set to the minimum distance of a neighbor involved in a
broadcast phase plus 1; see the macro $\compute(u)$.

\subsection{Typical Execution}\label{sub:normalexec}
Assume the system starts from a configuration where, for every process
$u$, $\status_u = C$. A process $u$ detecting an inconsistency in {\tt
  I} ({\em i.e.}, when $\neg\PICorrect(u)$ holds) stops executing {\tt
  I} and initiates a reset using $\ruleR(u)$, unless one of its neighbors
$v$ is already broadcasting a reset, in which case it joins the
broadcast of some neighbor by $\ruleRB(u)$.
To initiate a reset, $u$ sets $(\status_u,\dist_u)$ to $(RB,0)$
meaning that $u$ is the root of a reset (see macro $\beRoot(u)$), and
resets its {\tt I}'s variables to an pre-defined state of {\tt I},
which satisfies $\PReset(u)$, by executing the macro $\reset(u)$.
Whenever a process $v$ has a neighbor involved in a broadcast phase of
a reset (status $RB$), it stops executing {\tt I} and joins an
existing reset using $\ruleRB(v)$, even if its state in {\em I} is
correct, ({\em i.e.}, even if $\PICorrect(v)$ holds).  To join a
reset, $v$ also switches its status to $RB$ and resets its {\tt I}'s
variables ($\reset(v)$), yet it sets $\dist_v$ to the minimum distance
of its neighbors involved in a broadcast phase plus 1; see the macro
$\compute(v)$. Hence, if the configuration of {\tt I} is not
legitimate, then within at most $n$ rounds, each process receives the
broadcast of some reset. Meanwhile, processes (temporarily) stop
executing {\tt I} until the reset terminates in their closed
neighborhood thanks to the predicate $\PClean$.

When a process $u$ involved in the broadcast phase of some reset
realizes that all its neighbors are involved into a reset ({\em i.e.},
have status $RB$ or $RF$), it initiates the feedback phase by
switching to status $RF$, using $\ruleRF(u)$. The feedback phase is
then propagated up in the DAG described by the distance value: a
broadcasting process $u$ switches to the feedback phase if each of its
neighbors $v$ has not status $C$ and if $\dist_v > \dist_u$, then $v$
has status $RF$. This way the feedback phase is propagated up into the
DAG within at most $n$ additional rounds. Once a root of some reset
has status $RF$, it can initiate the last phase of the reset: all
processes involves into the reset has to switch to status $C$, using
$\ruleC$, meaning that the reset is done. The values $C$ are
propagated down into the reset DAG within at most $n$ additional
rounds. A process $u$ can executing {\tt I} again when all members of
its closed neighborhood (that is, including $u$ itself) have status
$C$, {\em i.e.}, when it satisfies $\PClean(u)$.

Hence, overall in this execution, the system reaches a configuration
$\gamma$ where all resets are done within at most $3n$ rounds. In
$\gamma$, all processes have status $C$. However, process has not
necessarily kept a state satisfying $\PReset$ ({\em i.e.}, the initial
pre-defined state of {\tt I}) in this configuration. Indeed, some
process may have started executing {\tt I} again before
$\gamma$. However, the predicate $\PClean$ ensures that no resetting
process has been involved in these latter (partial) executions of {\tt
  I}. Hence, \A rather ensures that all processes are in {\tt I}'s
states that are coherent with each other from $\gamma$. That is,
$\gamma$ is a so-called {\em normal configuration}, where $\PClean(u)
\And \PICorrect(u)$ holds for every process $u$.

\subsection{Stabilization of the Reset}\label{err:corr}

If a process $u$ is in an incorrect state of Algorithm \A ({\em i.e.},
if $\PRa(u) \Or \PRb(u)$ holds), we proceed as for inconsistencies in Algorithm {\tt
  I}. Either it joins an existing reset (using $\ruleRB(u)$) because at
least one of its neighbors is in a broadcast phase, or it initiates its
own reset using $\ruleR(u)$. 
Notice also that starting from an arbitrary configuration, the system
may contain some reset in progress. However, similarly to the typical
execution, the system stabilizes within at most $3n$ rounds to a normal
configuration.

Algorithm \A is also efficient in moves. Indeed, in
Sections~\ref{sect:alliance} and~\ref{sect:unison} we will give two
examples of composition {\tt I} $\circ$ \A that stabilize in a
polynomial number of moves. Such complexities are mainly due to the
coordination of the resets which, in particular, guarantees that if a
process $u$ is enabled to initiate a reset ($\PRUp(u)$) or the root of
a reset with status $RB$, then it satisfies this disjunction since the
initial configuration ({\em cf.}, Theorem~\ref{theo:pseudoRoots}, page
\pageref{theo:pseudoRoots}).

\subsection{Requirements on the Input Algorithm}\label{sect:require}

According to the previous explanation, Algorithm {\tt I} should satisfy the
following prerequisites:
\begin{enumerate}
\item Algorithm {\tt I} should not write into the variables of \A, {\em i.e.},
  variables $\status_u$ and $\dist_u$, for every process
  $u$. \label{RQ1}
\item For each process $u$, Algorithm {\tt I} should provide the two input
  predicates $\PICorrect(u)$ and $\PReset(u)$ to \A, and the macro
  $\reset(u)$. Those inputs should satisfy:
\begin{enumerate}
\item $\PICorrect(u)$ does not involve any variable of \A and is closed
  by Algorithm {\tt I}. \label{RQ2}
\item $\PReset(u)$ involves neither a variable of \A nor a variable of a neighbor of $u$. \label{RQ4}

\item If $\neg \PICorrect(u) \vee \neg \PClean(u)$ holds ({\em n.b.}
  $\PClean(u)$ is defined in \A), then no rule of Algorithm {\tt I} is enabled
  at $u$. \label{RQ3}

\item If $\PReset(v)$ holds, for every $v \in \N[u]$, then
  $\PICorrect(u)$ holds.\label{RQ6}
\item If $u$ performs a move in $\gamma \mapsto \gamma'$, where, in
  particular, it modifies its variables in Algorithm {\tt I} by
  executing $\reset(u)$ (only), then $\PReset(u)$ holds in $\gamma'$.\label{RQ5}
\end{enumerate}
\end{enumerate}

\begin{algorithm}
\small  
$~$ \\[0.3cm]
\textbf{Inputs:} \\[0.1cm]
\begin{tabular}{llll}
$\bullet$ & $\PICorrect(u)$ & : & predicate from the input algorithm {\tt I}\\
$\bullet$ & $\PReset(u)$ & : & predicate from the input algorithm {\tt I}\\
$\bullet$ & $\reset(u)$ & : & macro from the input algorithm {\tt I}
\end{tabular}
\\[0.3cm]
\textbf{Variables:} \\[0.1cm]
\begin{tabular}{llll}
$\bullet$ & $\st_u \in \{ C, RB, RF \}$ & : & the status of $u$\\
$\bullet$ & $\dist_u \in \mathds{N}$ & : & the  distance value associated to $u$
\end{tabular}
\\[0.3cm]
\textbf{Predicates:} \\[0.1cm]
\begin{tabular}{llll}
$\bullet$ & $\correct(u)$ & $\equiv$ & $\status_u =  C \Rightarrow \PICorrect(u)$\\
$\bullet$ & $\PClean(u)$  & $\equiv$ & $\forall v \in \N[u], \status_u  = C$\\
$\bullet$ & $\PRa(u)$     & $\equiv$ & $\status_u = C \And \neg\PReset(u) \And (\exists v \in \N(u)\ \mid\ \status_v  = RF)$\\
$\bullet$ & $\PB(u)$      & $\equiv$ & $\st_u = C \And (\exists v \in \N(u)\ |\ \st_v  = RB)$\\
$\bullet$ & $\PF(u)$      & $\equiv$ & $\st_u = RB \And \PReset(u) \And$\\ 
          &               &          & $(\forall v \in \N(u), (\st_v  = RB \And \dist_v \leq \dist_u) \Or (\status_v  = RF 
\And \PReset(v)))$\\
$\bullet$ & $\PC(u)$      & $\equiv$ & $\status_u = RF \And$\\
          &               &          & $(\forall v \in \N[u], \PReset(v) \And ((\status_v  = RF \And
\dist_v \geq \dist_u) \Or (\status_v = C)))$\\
$\bullet$ & $\PRb(u)$     & $\equiv$ & $\status_u \neq C \And \neg\PReset(u)$ \\
$\bullet$ & $\PRUp(u)$    & $\equiv$ & $\neg\PB(u) \And (\PRa(u) \Or \PRb(u) \Or \neg\correct(u))$
\end{tabular}  
\\[0.3cm]
\textbf{Macros:}\\[0.1cm]
\begin{tabular}{llll}
$\bullet$ & $\beRoot(u)$  & : & $\status_u := RB$;  $\dist_u := 0$; \\
$\bullet$ & $\compute(u)$ & : & $\status_u := RB$; $\dist_u := \argmin_{(v \in \N(u) ~ \And ~\st_v = RB)}
(\dist_v)+1$; 
\end{tabular}  
\\[0.3cm]
{\textbf{Rules:}} \\[0.1cm]
\begin{tabular}{lllll}
$\ruleRB(u)$ & : & $\PB(u)$   & $\to$ & $\compute(u)$; $\reset(u)$;\\
$\ruleRF(u)$ & : & $\PF(u)$   & $\to$ & $\status_u := RF$; \\
$\ruleC(u)$  & : & $\PC(u)$   & $\to$ & $\status_u := C$; \\
$\ruleR(u)$  & : & $\PRUp(u)$ & $\to$ & $\beRoot(u)$;  $\reset(u)$;
\end{tabular}
\\[0.3cm]
\caption{Algorithm \A, code for every process $u$}
\label{alg:A}
\end{algorithm}

\section{Correctness and Complexity Analysis}

\subsection{Partial Correctness}

\begin{lemma}\label{lem:nonPRabc}
In any terminal configuration of \A,  $\neg \PRa(u) \And \neg \PRb(u) \And \correct(u)$ holds for every process $u$.
\end{lemma}
\begin{proof}
  Let $u$ be any process and consider any terminal configuration of
  \A.  Since $\ruleRB(u)$ and $\ruleR(u)$ are disabled, $\neg\PB(u)$
  and $\neg\PRUp(u)$ hold. Since $\neg\PB(u) \And \neg\PRUp(u)$
  implies $\neg \PRa(u) \And \neg \PRb(u) \And \correct(u)$, we are
  done.
\end{proof}

Since $\neg \PRb(u) \equiv \status_u = C \Or \PReset(u)$, we have the
following corollary.

\begin{corollary}\label{coro:nonPRB}
  In any terminal configuration of \A, $\status_u = C \Or \PReset(u)$
  holds for every process $u$.
\end{corollary}

\begin{lemma}\label{lem:nonRB}
  In any terminal configuration of \A, $\status_u \neq RB$ for every
  process $u$.
\end{lemma}
\begin{proof}
  Assume, by the contradiction, that some process $u$ satisfies
  $\status_u = RB$ in a terminal configuration of \A. Without the loss
  of generality, assume $u$ is a process such that $\status_u = RB$
  with $\dist_u$ maximum. First, $\PReset(u)$ holds by
  Corollary~\ref{coro:nonPRB}.
  Then, every neighbor $v$ of $u$ satisfies $\status_v \neq C$, since
  otherwise $\ruleRB(v)$ is enabled.  So, every $v$ satisfies
  $\status_v \in \{RB,RF\}$, $\status_v = RB \Rightarrow dist_v \leq
  \dist_u$ (by definition of $u$), and $\status_v = RF \Rightarrow
  \PReset(v)$ (by Corollary~\ref{coro:nonPRB}). Hence, $\ruleRF(u)$ is
  enabled, a contradiction.
\end{proof}

\begin{lemma}\label{lem:nonRF}
  In any terminal configuration of \A, $\status_u \neq RF$ for every
  process $u$.
\end{lemma}
\begin{proof}
  Assume, by the contradiction, that some process $u$ satisfies
  $\status_u = RF$ in a terminal configuration of \A. Without the loss
  of generality, assume $u$ is a process such that $\status_u = RF$
  with $\dist_u$ minimum. First, every neighbor $v$ of $u$ satisfies
  $\status_v \neq RB$, by Lemma~\ref{lem:nonRB}. Then, every neighbor
  $v$ of $u$ such that $\status_v = C$ also satisfies $\PReset(v)$,
  since otherwise $\PRa(v)$ holds, contradicting then
  Lemma~\ref{lem:nonPRabc}. Finally, by definition of $u$ and by
  Corollary~\ref{coro:nonPRB}, every neighbor $v$ of $u$ such that
  $\status_v = RF$ both satisfies $dist_v \geq \dist_u$ and
  $\PReset(v)$. Hence, $\ruleC(u)$ is enabled, a contradiction.
\end{proof}

\begin{theorem}\label{theo:termA}
  For every configuration $\gamma$ of \A, $\gamma$ is terminal if and
  only if $\PClean(u)
  \And \PICorrect(u)$ holds in $\gamma$, for every process $u$.
\end{theorem}
\begin{proof}
  Let $u$ be any process and assume $\gamma$ is terminal. By
  Lemmas~\ref{lem:nonRB} and~\ref{lem:nonRF}, $\status_u =C$ holds in
  $\gamma$. So, $\PClean(u)$ holds in $\gamma$. Moreover, since
  $\correct(u)$ holds (Lemma~\ref{lem:nonPRabc}), $\PICorrect(u)$ also
  holds in $\gamma$, and we are done.

  Assume now that for every process $u$, $\PClean(u) \And
  \PICorrect(u)$ holds in $\gamma$. Then, $\status_u =C$ for every
  process $u$, and so $\ruleC(u)$, $\ruleRB(u)$, and $\ruleRF(u)$ are
  disabled for every $u$. Then, since every process has status $C$, 
  $\neg \PRa(u) \And \neg \PRb(u)$ holds, moreover, $\PICorrect(u)$ implies
  $\correct(u)$, so $\ruleR(u)$ is also disabled in $\gamma$. Hence $\gamma$ is
  terminal, and we are done.
\end{proof}

\subsection{Termination}

From Requirements~\ref{RQ1} and~\ref{RQ2}, we know that Algorithm
{\tt I} does not write into $\status_u$ and $\PICorrect(u)$ is closed
by {\tt I}, for every process $u$. Hence follows.

\begin{remark}\label{rem:cor:close:I}
  For every process $u$, predicate $\correct(u)$ (defined in \A) is
  closed by  {\tt I}.
\end{remark}

Requirements~\ref{RQ1},~\ref{RQ4}, and~\ref{RQ3} ensures the
following property.

\begin{lemma}\label{lem:closedI}
  For every process $u$, predicates $\neg\PRa(u)$, $\neg\PRb(u)$, and
  $\PB(u)$ are closed by {\tt I}.
\end{lemma}
\begin{proof}
  Let $\gamma \mapsto \gamma'$ be any step of ${\tt I}$.

\begin{itemize}
\item Assume that  $\neg\PRa(u)$ holds at some process $p$ in $\gamma$. 

  If $\status_u \neq C \vee (\forall v \in \N(u)\ \mid\ \status_v \neq
  RF)$ in $\gamma$, then $\status_u \neq C \vee (\forall v \in \N(u)\
  \mid\ \status_v \neq RF)$ still holds in $\gamma'$ by Requirement~\ref{RQ1}, and we are done.

  Otherwise, $\status_u = C \And \PReset(u) \And (\exists v \in \N(u)\
  \mid\ \status_v = RF)$ holds in $\gamma$. 
In particular, $\neg
  \PClean(u)$ holds in $\gamma$. Hence,
 no rule of {\tt I} is enabled at $u$ in $\gamma$, by
  Requirement~\ref{RQ3}. Consequently, $\PReset(u)$ still holds in $\gamma'$. Since, $\PReset(u)$ implies $\neg\PRa(u)$, we are done.

\item Assume that $\neg\PRb(u)$ holds at some process $u$ in $\gamma$.
  If $\status_u = C$ holds in $\gamma$, then $\status_u = C$ holds in
  $\gamma'$ by Requirement~\ref{RQ1}, and so $\neg\PRb(u)$ still holds
  in $\gamma'$. Otherwise, $\status_u \neq C \wedge \PReset(u)$ holds
  in $\gamma$. In particular, $\neg \PClean(u)$ holds in $\gamma$.
  Hence, no rule of {\tt I} is enabled at $u$, by
  Requirement~\ref{RQ3}, and by Requirement \ref{RQ4}, $\PReset(u)$,
  and so $\neg\PRb(u)$, still holds in $\gamma'$.

\item By Requirement~\ref{RQ1}, $\PB(u)$ is closed by {\tt I}.
\end{itemize}
\end{proof}

Recall that two rules are {\em mutually exclusive} if there is no
configuration $\gamma$ and no process $u$ such that both rules at
enabled $u$ in $\gamma$.  Two algorithms are {\em mutually exclusive}
if their respective rules are pairwise mutually exclusive.  Now,
whenever a process $u$ is enabled in \A, $\neg \PICorrect(u) \vee \neg
\PClean(u)$ holds and, by Requirement \ref{RQ3}, no rule of {\tt I} is
enabled at $u$. Hence, follows.

\begin{remark}\label{rem:mutex}
  Algorithms \A and {\tt I} are mutually exclusive.
\end{remark}

\begin{lemma}\label{lem:mutex:2}
  Rules of Algorithm \A are pairwise mutually exclusive.
\end{lemma}
\begin{proof}
  Since $\PB(u)$ implies $\status_u = C$, $\PF(u)$ implies $\status_u
  = RB$, and $\PC(u)$ implies $\status_u = RF$, we can conclude that
  $\ruleRB(u)$, $\ruleRF(u)$, and $\ruleC(u)$ are pairwise mutually
  exclusive. 

  Then, since $\PRUp(u)$ implies $\neg \PB(u)$, rules $\ruleR(u)$ and
  $\ruleRB(u)$ are mutually exclusive.

  $\PC(u)$ implies $\correct(u) \wedge \PReset(u)$ which, in turn,
  implies $\neg \PRUp(u)$. Hence, $\ruleR(u)$ and $\ruleC(u)$ are
  mutually exclusive.

  $\PF(u)$ implies $\status_u = RB \wedge \PReset(u)$. Now, $\PRUp(u)$
  implies $\status_u = C \vee \neg\PReset(u)$.  Hence, $\ruleR(u)$ and
  $\ruleRF(u)$ are mutually exclusive.
\end{proof}

\begin{lemma}
\label{lem:PRa}
For every process $u$, predicates $\neg\PRa(u)$ and $\neg\PRb(u)$ are
closed by {\tt I} $\circ$ \A.
\end{lemma}
\begin{proof}
  By Remark~\ref{rem:mutex} and Lemma~\ref{lem:closedI}, to prove this
  lemma it is sufficient to show that $\neg\PRa(u)$ and $\neg\PRb(u)$
  are closed by \A, for every process $u$.

  Predicate $\neg\PRb(u)$ only depends on variables of $u$ by
  Requirement \ref{RQ4}. So, if $u$ does not move, $\neg\PRb(u)$ still
  holds. Assume $\neg\PRb(u)$ holds in $\gamma$ and $u$ executes a
  rule of \A in $\gamma \mapsto \gamma'$. If $u$ executes $\ruleRB(u)$
  or $\ruleR(u)$, then $u$ modifies its variables in {\tt I} by executing
  $\reset(u)$. Hence, in both cases, $\PReset(u)$ holds in $\gamma'$
  by Requirement \ref{RQ5} and as $\PReset(u)$ implies $\neg\PRb(u)$,
  we are done.  Otherwise, $u$ executes $\ruleRF(u)$ or
  $\ruleC(u)$. In both cases, $\PReset(u)$ holds in $\gamma$ and so in
  $\gamma'$ by Requirement \ref{RQ4}, and we are done.

  Assume now that the predicate $\neg\PRa(u)$ holds in $\gamma$ and
  consider any step $\gamma \mapsto \gamma'$. Assume first that $u$
  moves in $\gamma \mapsto \gamma'$. If $u$ executes $\ruleRB(u)$,
  $\ruleRF(u)$, or $\ruleR(u)$, then $\status_u \neq C$ in $\gamma'$,
  hence $\neg\PRa(u)$ holds in $\gamma'$. If $u$ executes $\ruleC(u)$
  in $\gamma \mapsto \gamma'$, $u$ satisfies $\PReset(u)$ in $\gamma$,
  and so in $\gamma'$ by Requirement \ref{RQ4}. Since $\PReset(u)$
  implies $\neg\PRa(u)$, we are done. Assume now that $u$ does not
  move in $\gamma \mapsto \gamma'$. In this case, $\PRa(u)$ may become
  true only if at least a neighbor $v$ of $u$ switches to status $RF$,
  by executing $\ruleRF(v)$. Now, in this case, $\status_u \neq C$ in
  $\gamma$, and so in $\gamma'$. Consequently, $\neg\PRa(u)$ still
  holds in $\gamma'$.
\end{proof}

\begin{theorem}
\label{theo:correct}
For every process $u$, $\correct(u)$ $\Or$ $\PB(u)$ is closed by {\tt
  I} $\circ$ \A.
\end{theorem}
\begin{proof}
  By Remarks~\ref{rem:cor:close:I} and~\ref{rem:mutex}, and
  Lemma~\ref{lem:closedI}, to prove this lemma it is sufficient to
  show that $\correct(u)$ $\Or$ $\PB(u)$ is closed by \A, for every
  process $u$.

  Let $\gamma \mapsto \gamma'$ be any step of \A such that
  $\correct(u)$ $\Or$ $\PB(u)$ holds in $\gamma$.
  \begin{itemize}
  \item Assume $\correct(u)$ holds in $\gamma$. By
    Requirement~\ref{RQ2}, if $\PICorrect(u)$ holds in $\gamma$, then
    $\PICorrect(u)$ still holds in $\gamma'$, and as $\PICorrect(u)$
    implies $\correct(u)$, we are done.

    Assume now $\neg \PICorrect(u)$ holds in $\gamma$. Then,
    $\correct(u) \wedge \neg \PICorrect(u)$ implies $\status_u \neq C$
    in $\gamma$. Since $\PC(u)$ implies $\PICorrect(u)$ by
    Requirement~\ref{RQ6}, $\ruleC(u)$ is disabled in $\gamma$, and
    Consequently, $\status_u \neq C$ in $\gamma'$, which implies that
    $\correct(u)$ still holds in $\gamma'$.

  \item Assume $\PB(u)$ holds in $\gamma$. If $u$ moves in $\gamma
    \mapsto \gamma'$, then $u$ necessarily executes $\ruleRB(u)$; see
    Lemma~\ref{lem:mutex:2}. In this case, $\status_u = RB$ in
    $\gamma'$, which implies $\correct(u)$ in $\gamma'$.

    If $u$ does not move, then at least one neighbor of $u$ should
    switch its status from $RB$ to either $C$ or $RF$ so that $\neg
    \PB(u)$ holds in $\gamma'$. Any neighbor $v$ of $u$ satisfying
    $\status_v = RB$ may only change its status by executing
    $\ruleRF(v)$ in $\gamma \mapsto \gamma'$. Now, $\ruleRF(v)$ is
    necessarily disabled in $\gamma$ since $\status_u = C$. Hence,
    $\PB(u)$ still holds in $\gamma'$ in this case.
  \end{itemize}
\end{proof}

From Lemma \ref{lem:PRa} and Theorem \ref{theo:correct}, we can deduce
the following corollary.
\begin{corollary}
\label{cor:correct}
For every process $u$, $\neg\PRUp(u)$ is closed  by {\tt I} $\circ$ \A.
\end{corollary}

\subsubsection{Roots.}

If the configuration is illegitimate {\em w.r.t.} the initial
algorithm, then some processes locally detect the inconsistency by
checking their state and that of their neighbors (using Predicate
$\PICorrect$).  Such processes, called here {\em roots}, should
initiates a reset. Then, each root $u$ satisfies $\status_u \neq C$
all along the reset processing. According to its status, a root is
either {\em alive} or {\em dead}, as defined below.

\begin{definition} Let $\PRoot(u)$ $ \equiv $ $\status_u = RB$ $\And$ $(\forall v \in \N(u)$,  $\status_v = RB \Rightarrow \dist_v \geq \dist_u)$.
\begin{itemize}
\item A process $u$ is said to be an {\em alive root} if $\PRUp(u) \Or \PRoot(u)$.
\item A process $u$ is said to be an {\em dead root} if 
$\status_u = RF \And (\forall v \in \N(u), status_v \neq C \Rightarrow \dist_v \geq \dist_u)$.
\end{itemize}
\end{definition}

By definition, follows.

\begin{remark}\label{rem:dead}
For every process $u$, if  $\PC(u)$ holds, then $u$ is a dead root.
\end{remark}

The next theorem states that no alive root is created during an
execution.

\begin{theorem}
\label{theo:pseudoRoots}
For every process $u$, $\neg\PRoot(u) \And \neg\PRUp(u)$ is closed  by {\tt I} $\circ$ \A.
\end{theorem}
\begin{proof}
By Requirement~\ref{RQ1} and~Corollary~\ref{theo:correct},
$\neg\PRoot(u) \And \neg\PRUp(u)$ is closed by {\tt I}. Hence, by
Remark~\ref{rem:mutex}, it is sufficient to show that $\neg\PRoot(u)
\And \neg\PRUp(u)$ is closed by \A.

 Let $\gamma \mapsto \gamma'$ be any step of \A such that
 $\neg\PRoot(u) \And \neg\PRUp(u)$  holds in $\gamma$.
 By Corollary~\ref{theo:correct}, $\neg\PRUp(u)$ holds in $\gamma'$.
 To show that  $\neg\PRoot(u)$ holds in $\gamma'$, we now consider the following cases:
 \begin{description}
 \item[$\status_u = RF$ in $\gamma$:] In this case, $\ruleRB(u)$ and
   $\ruleR(u)$ are respectively disabled in $\gamma$ since $\status_u
   \neq C$ and $\neg\PRUp(u)$ hold in $\gamma$. So, $\status_u \neq
   RB$ in $\gamma'$, which implies that $\neg\PRoot(u)$ still holds in
   $\gamma'$.

 \item[$\status_u = RB$ in $\gamma$:] Then, $\neg\PRoot(u)$ in
   $\gamma$ implies that there is a neighbor $v$ of $u$ such that
   $\status_v = RB \And \dist_v < \dist_u$ in $\gamma$. Let $\alpha$
   be the value of $\dist_v$ in $\gamma$.  Since $\neg\PRUp(u)$ holds
   in $\gamma$, $u$ may only execute $\ruleRF(u)$ in $\gamma \mapsto
   \gamma'$ and, consequently, $\dist_u > \alpha$ in $\gamma'$. Due to
   the value of $\dist_u$, $v$ may only execute $\ruleR(v)$ in $\gamma
   \mapsto \gamma'$. Whether or not $v$ moves, $\status_v = RB \wedge
   \dist_v \leq \alpha$ in $\gamma'$. Hence, $\status_v = RB \wedge
   \dist_v < \dist_u$ in $\gamma'$, which implies that $\neg\PRoot(u)$
   still holds in $\gamma'$.

 \item[$\status_u = C$ in $\gamma$:] If $u$ does not moves in $\gamma
   \mapsto \gamma'$, then $\neg\PRoot(u)$ still holds in
   $\gamma'$. Otherwise, since $\neg\PRUp(u)$ holds in $\gamma$, $u$
   can only execute $\ruleRB(u)$ in $\gamma \mapsto \gamma'$. In this
   case, $\PB(u)$ implies that there is a neighbor $v$ of $u$ such
   that $\status_v = RB$ in $\gamma$. Without the loss of generality,
   assume $v$ is the neighbor of $u$ such that $\status_v = RB$ with
   the minimum distance value in $\gamma$. Let $\alpha$ be the value
   of $\dist_v$ in $\gamma$. Then, $\status_u = RB$ and $\dist_u =
   \alpha+1$ in $\gamma'$. Moreover, since $\status_u = C$ and
   $\status_v = RB$ in $\gamma$, $v$ may only execute $\ruleR(v)$ in
   $\gamma \mapsto \gamma'$. Whether or not $v$ moves, $\status_v = RB
   \wedge \dist_v \leq \alpha$ in $\gamma'$.  Hence, $\status_v = RB
   \wedge \dist_v < \dist_u$ in $\gamma'$, which implies that
   $\neg\PRoot(u)$ still holds in $\gamma'$.
 \end{description}
\end{proof}

\subsubsection{Move Complexity.}

\begin{definition}[$AR$]
  Let $\gamma$ be a configuration of {\tt I} $\circ$ \A. We denote by
  $AR(\gamma)$ the set of alive roots in $\gamma$.
\end{definition}

By Theorem~\ref{theo:pseudoRoots}, follows.

\begin{remark}\label{rem:inclusion}
Let $\gamma_0\cdots\gamma_i\cdots$ be any execution of  {\tt I} $\circ$ \A.
For every $i > 0$, $AR(\gamma_i) \subseteq AR(\gamma_{i-1})$.
\end{remark}

Based on the aforementioned property, we define below the notion of
{\em segment}.

\begin{definition}[Segment] Let $e = \gamma_0\cdots\gamma_i\cdots$ be any execution of  {\tt I} $\circ$ \A.
\begin{itemize}
\item If for every $i > 0$, $|AR(\gamma_{i-1})| = |AR(\gamma_i)|$,
  then the first {\em segment} of $e$ is $e$ itself, and there is no
  other segment.
\item Otherwise, let $\gamma_{i-1} \mapsto \gamma_i$ be the first step
  of $e$ such that $|AR(\gamma_{i-1})| > |AR(\gamma_i)|$. The {\em
    first segment} of $e$ is the prefix $\gamma_0\cdots\gamma_i$ and
  the {\em second segment} of $e$ is the first segment of the suffix
of $e$ starting in $\gamma_i$, and so forth.
\end{itemize}
\end{definition}

By Remark~\ref{rem:inclusion}, follows.
\begin{remark}\label{rem:nbseg}
  Every execution of {\tt I} $\circ$ \A contains at most $n+1$
  segments where $n$ is the number of processes.
\end{remark}

We now study how a reset propagates into the network. To that goal, we first define the notion of {\em reset parent}. Roughly speaking, the parents of $u$ in a reset are its neighbors (if any) that have caused its reset.

\begin{definition}[Reset Parent and Children]
  $RParent(v,u)$ holds for any two processes $u$ and $v$ if $v \in
  \N(u)$, $\status_u \neq C$, $\PReset(u)$, $\dist_u > \dist_v$, and
  $(\status_u = \status_v \vee \status_v = RB)$. 

  Whenever $RParent(v,u)$ holds, $v$
  (resp., $u$) is said to be a {\em reset parent} of $u$ in (resp., a
  {\em reset child} of $v$).
\end{definition}

Remark that in a given configuration, a process may have several reset
parents.  Below, we define the {\em reset branches}, which are the trails of
a reset in the network.

\begin{definition}[Reset Branch]
  A \emph{reset branch} is a 
   sequence of processes~$u_1, \ldots ,u_k$ for some integer~$k \geq
  1$, such that~$u_1$ is~an alive or dead root and, for every~$1
  < i \leq k$, we have~$RParent(u_{i-1},u_i)$.
  The process~$u_i$ is said to be at \emph{depth}~$i-1$ and $u_i, \cdots
  ,u_k$ is called a {\em reset sub-branch}. The process~$u_1$ is the {\em initial extremity} of the reset branch~$u_1, \ldots ,u_k$.
\end{definition}

\begin{lemma}\label{lem:trace}
Let $u_1, \ldots ,u_k$ be any reset branch. 
\begin{enumerate} 
\item $k \leq n$, \label{lem:trace:1}
\item If $\status_{u_1} = C$, then $k = 1$. Otherwise, $\status_{u_1}\cdots\status_{u_k} \in RB^*RF^*$.\label{lem:trace:2}
\item $\forall i \in \{2, \ldots, k\}$, $u_i$ is neither an alive, not
  a dead root.\label{lem:trace:3}
\end{enumerate} 
\end{lemma}
\begin{proof} Let $i$ and $j$ such that $1 \leq i < j \leq
  k$.  By definition, $\dist_{u_i} < \dist_{u_j}$ and so $u_i \neq
  u_j$. Hence, in a reset branch, each node appears at most once, and
 Lemma~\ref{lem:trace}.\ref{lem:trace:1} holds.

 Let $i \in \{2, \ldots, k\}$.  Lemma~\ref{lem:trace}.\ref{lem:trace:2} immediately
 follows from the following three facts, which directly derive from
 the definition of reset parent.
\begin{itemize}
\item $\status_{u_i} \neq C$.
\item $\status_{u_i} = RB \Rightarrow
  \status_{u_{i-1}} = RB$.
\item $\status_{u_i} = RF \Rightarrow
  \status_{u_{i-1}} \in \{RB,RF\}$.
\end{itemize}

Lemma~\ref{lem:trace}.\ref{lem:trace:3} immediately follows from those
two facts.
\begin{itemize}
\item {\em $u_i$ is not a dead root}, since $u_{i-1} \in \N(u_i) \wedge \status_{u_{i-1}} \neq C \wedge \dist_{u_{i-1}} < \dist_{u_i}$.
\item {\em $u_i$ is not an alive root}, indeed
\begin{itemize}
\item {\em $\neg \PRoot(u)$ holds}, since $u_{i-1} \in \N(u_i) \wedge (\status_{u_i} = RB \Rightarrow
  \status_{u_{i-1}} = RB) \wedge \dist_{u_{i-1}} < \dist_{u_i}$.
\item {\em $\neg \PRUp(u_i)$ holds} since $\neg \PRa(u) \And \neg
\PRb(u)$ holds because $\PReset(u)$ holds, and $\correct(u)$ holds
because $\status_u \neq C$.
\end{itemize}
\end{itemize}

\end{proof}

\begin{remark}\label{rem:unless}
In a configuration, a process $u$ may belong to several
branches. Precisely, $u$ belongs to at least one reset branch, unless
$\status_u = C \wedge \PICorrect(u)$ holds.
\end{remark}

\begin{lemma}\label{still:alive}
Let $\gamma_x \mapsto \gamma_{x+1}$ be a step of {\tt I} $\circ$ \A. Let $u_1, \ldots, u_k$ be a  reset branch in $\gamma_x$. If $u_1$ is an alive root in $\gamma_{x+1}$, then  $u_1, \ldots, u_k$ is a  reset branch in $\gamma_{x+1}$.
\end{lemma}
\begin{proof}
We first show the following two claims:
\begin{description}
\item[Claim 1:] {\em If $u_1$ moves in $\gamma_x \mapsto \gamma_{x+1}$, then $u_1$ necessarily executes $\ruleR$ in $\gamma_x \mapsto \gamma_{x+1}$.}

  {\em Proof of the claim:} Since $u_1$ is an alive root in
  $\gamma_{x+1}$, $u_1$ is an alive root in $\gamma_x$, by
  Theorem~\ref{theo:pseudoRoots}. Assume now, by the contradiction,
  that $u_1$ moves, but does not execute $\ruleR$ in $\gamma_x \mapsto
  \gamma_{x+1}$. Then, $\neg \PRUp(u_k)$ holds in $\gamma_x$, by
  Remark~\ref{rem:mutex} and Lemma~\ref{lem:mutex:2}. So, by
  definition of alive root, $\status_u = RB$ in $\gamma_x$ and, from
  the code of \A and Requirement~\ref{RQ3}, $u_1$ executes $\ruleRF$
  in $\gamma_x \mapsto \gamma_{x+1}$. Consequently, $\status_u = RF$. Now, $\neg \PRUp(u_k)$ still holds in $\gamma_{x+1}$, by
  Corollary~\ref{cor:correct}. Hence, $u_1$ is not an alive root in
  $\gamma_{x+1}$, a contradiction.

\item[Claim 2:] {\em For every $i \in \{2, \ldots, k\}$, if $u_i$ moves
  $\gamma_x \mapsto \gamma_{x+1}$, then $u_i$ executes $\ruleRF$ in
  $\gamma_x \mapsto \gamma_{x+1}$ and in $\gamma_x$ we have
  $\status_{u_i} = RB$ and $i < k \Rightarrow \status_{u_{i+1}} = RF$.}

  {\em Proof of the claim:} We first show that only $\ruleRF$ may be enabled at 
$u_i$ in $\gamma_x$.
\begin{itemize}
\item By definition, $\status_{u_i} \neq C$ and so $\neg
  \PClean(u_i)$ holds in $\gamma_x$. Thus, by Requirement~\ref{RQ3}, all
  rules of {\tt I} that are disabled at $u_i$ in $\gamma_x$.
\item $\ruleRB(u_i)$ is disabled in $\gamma_x$ since
  $\status_{u_i} \neq C$ (by definition).
\item The fact that $u_i$ is not a dead root in $\gamma_x$
  (Lemma~\ref{lem:trace}) implies that $\ruleC(u_i)$ is disabled 
  in $\gamma_x$ (Remark~\ref{rem:dead}).
\item $\ruleR(u_i)$ is disabled in $\gamma_x$ since $u_i$ is not an alive root
  (Lemma~\ref{lem:trace}).
\end{itemize}
Hence, $u_i$ can only executes $\ruleRF$ in $\gamma_x \mapsto
\gamma_{x+1}$. In this case, $\status_{u_i} = RB$ in $\gamma_x$; see
the guard of $\ruleRF(u_i)$. Moreover, if $i < k$, then
$\status_{u_i} = RB \wedge RParent(u_i,u_{i+1})$ implies that $u_{i+1} \in \N(u_i)$,
$\status_{u_{i+1}} \in \{RB,RF\}$, and $\dist_{u_i} <
\dist_{u_{i+1}}$. Now, if $u_{i+1} \in \N(u_i)$, $\status_{u_{i+1}} =
RB$, and $\dist_{u_i} < \dist_{u_{i+1}}$, then $\ruleRF(u_i)$ is
disabled. Hence, if $u_i$ moves $\gamma_x \mapsto \gamma_{x+1}$ and $i
< k$, then $u_i$ executes $\ruleRF$ and so $\status_{u_{i+1}} = RF$ in
$\gamma_x$.
\end{description}
Then, we proceed by induction on $k$. The base case ($k = 1$) is
trivial. Assume now that $k > 1$. Then, by induction hypothesis, $u_1,
\ldots, u_{k-1}$ is a reset branch in $\gamma_{x+1}$. Hence, to show
that $u_1, \ldots, u_k$ is a reset branch in $\gamma_{x+1}$, it is
sufficient to show that $RParent(u_{k-1},u_k)$ holds in
$\gamma_{x+1}$. Since, $RParent(u_{k-1},u_k)$ holds in
$\gamma_x$, we have $\status_{u_k} \neq C$ in $\gamma_x$. So, we now study the following two cases:

\begin{itemize}
\item $\status_{u_k} = RB$ in $\gamma_x$. By Claim 2, $u_k$ may only
  execute $\ruleRF$ in $\gamma_x \mapsto \gamma_{x+1}$. Consequently,
  $\status_{u_k} \in \{RB,RF\} \wedge \PReset(u_k) \wedge \dist_{u_k}
  = d$ holds in $\gamma_{x+1}$, where $d > 0$ is the value of
  $\dist_{u_k}$ in $\gamma_x$.
  Consider now process $u_{k-1}$. Since $\status_{u_k} = RB$ in
  $\gamma_x$, $\status_{u_{k-1}} = RB$ too in $\gamma_x$
  (Lemma~\ref{lem:trace}).

  If $u_{k-1}$ does not move in $\gamma_x \mapsto \gamma_{x+1}$, then
  $RParent(u_{k-1},u_k)$ still holds in $\gamma_{x+1}$, and we are
  done.

  Assume now that $u_{k-1}$ moves in $\gamma_x \mapsto \gamma_{x+1}$.
  Then, we necessarily have $k = 2$ since otherwise, Claim 2 applies
  for $i = k-1$: $u_{k-1}$ moves in $\gamma_x \mapsto \gamma_{x+1}$
  only if $\status_{u_k} = RF$ in $\gamma_x$.  Now, $k = 2$ implies
  that $u_{k-1}$ executes $\ruleR$ in $\gamma_x \mapsto \gamma_{x+1}$
  (by Claim 1), so $\status_{u_{k-1}} = RB$ and $\dist_{u_{k-1}} = 0 <
  d$ in $\gamma_{x+1}$. Consequently, $RParent(u_{k-1},u_k)$ still
  holds in $\gamma_{x+1}$, and we are done.

\item $\status_{u_k} = RF$ in $\gamma_x$. 

By Claim 2, $u_k$ does not move in $\gamma_x \mapsto \gamma_{x+1}$. 
So, $\status_{u_k} = RF \wedge \PReset(u_k) \wedge \dist_{u_k} = d$ holds in $\gamma_{x+1}$, where $d > 0$ is the value of $\dist_{u_k}$
  in $\gamma_x$.

  If $u_{k-1}$ does not move in $\gamma_x \mapsto \gamma_{x+1}$, we
  are done. Assume, otherwise, that $u_{k-1}$ moves in $\gamma_x
  \mapsto \gamma_{x+1}$. 

  Then, if $k = 2$, then $u_{k-1}$ executes
  $\ruleR$ in $\gamma_x \mapsto \gamma_{x+1}$ (by Claim 1), so
  $\status_{u_{k-1}} = RB$ and $\dist_{u_{k-1}} = 0 < d$ in
  $\gamma_{x+1}$, and so $RParent(u_{k-1},u_k)$ still holds in
  $\gamma_{x+1}$. 

  Otherwise ($k > 2$),
  $u_{k-1}$ necessarily executes $\ruleRF$ in $\gamma_x \mapsto
  \gamma_{x+1}$ (by Claim 2): $\status_{u_{k-1}} = RF$ and $\dist_{u_{k-1}} <
  \dist_{u_k}$ in $\gamma_{x+1}$ ({\em n.b.}, neither
$\dist_{u_{k-1}}$ nor $\dist_{u_k}$ is modified in $\gamma_x \mapsto
\gamma_{x+1}$). Hence, $RParent(u_{k-1},u_k)$ still holds in
$\gamma_{x+1}$, and we are done.

\end{itemize}
\end{proof}

\begin{lemma}\label{lem:rfSeg}
  Let $u$ be any process. During a segment $S =
  \gamma_i\cdots\gamma_j$ of execution of {\tt I} $\circ$ \A, if $u$
  executes the rule~$\ruleRF$, then $u$ does not execute any other
  rule of \A in the remaining of $S$.
\end{lemma}
\begin{proof}
  Let~$\gamma_x \mapsto \gamma_{x+1}$ be a step of $S$ in which $u$
  executes~$\ruleRF$. Let~$\gamma_y \mapsto \gamma_{y+1}$ (with $y >
  x$) be the next step in which $u$ executes its next rule of \A.  (If
  $\gamma_x \mapsto \gamma_{x+1}$ or $\gamma_y \mapsto \gamma_{y+1}$
  does not exist, then the lemma trivially holds.)
  Then, since $\ruleRF(u)$ is enabled in $\gamma_x$, $\neg \PRUp(u)$
  holds in $\gamma_x$, by Lemma~\ref{lem:mutex:2}. Consequently, $\neg
  \PRUp(u)$ holds forever from $\gamma_x$, by
  Corollary~\ref{cor:correct}.  Hence, from the code of \A and
  Requirement~\ref{RQ3}, $u$ necessarily executes~$\ruleC$ in
  $\gamma_y \mapsto \gamma_{y+1}$ since $\status_u = RF \wedge \neg
  \PRUp(u)$ holds in $\gamma_y$.
  In $\gamma_x$, since $\status_u = RB$, $u$ belongs to some reset branches (Remark~\ref{rem:unless}) and 
all reset branches containing~$u$ have an alive root (maybe
  $u$) of status $RB$ (Lemma~\ref{lem:trace}). 
Let $v$ be any  alive
  root belonging to a reset branch containing~$u$ in $\gamma_x$.  
In
  $\gamma_y$, $u$ is the dead root, since $\PC(u)$ holds (Remark~\ref{rem:dead}). 
By Lemma~\ref{lem:trace}, either $u = v$ or $u$ no more belong to a reset branch whose initial extremity is $v$. 
By Lemma~\ref{still:alive} and  Theorem~\ref{theo:pseudoRoots}, $v$ is no more an alive root in  $\gamma_y$.
  Still by Theorem~\ref{theo:pseudoRoots}, the number of alive roots
  necessarily decreased between $\gamma_x$ and $\gamma_y$: $\gamma_x
  \mapsto \gamma_{x+1}$ and $\gamma_y \mapsto \gamma_{y+1}$ belong to
  two distinct segments of the execution.
\end{proof}

\begin{theorem}\label{theo:seg}
  The sequence of rules of \A executed by a process $u$ in a
  segment of execution of {\tt I} $\circ$ \A belongs to the following language:
$$(\ruleC + \varepsilon)\ (\ruleRB + \ruleR + \varepsilon)\ (\ruleRF + \varepsilon)$$
\end{theorem}
\begin{proof}
  From the code of \A and Requirement~\ref{RQ3}, we know that after
  any execution of $\ruleC(u)$, the next rule of \A $u$ will execute
  (if any), is either $\ruleRB$ or $\ruleR$. Similarly, immediately
  after an execution of $\ruleRB(u)$ (resp., $\ruleR(u)$), $\status_u
  = RB \wedge \PReset(u)$ holds (see Requirement~\ref{RQ5}) and
  $\PReset(u)$ holds while $u$ does not switch to status $C$
  (Requirements~\ref{RQ4} and \ref{RQ3}). So the next rule of \A $u$
  will execute (if any) is $\ruleRF$. Finally, immediately after any
  execution of $\ruleRF(u)$, $\status_u = RF \wedge \PReset(u)$ holds
  until (at least) the next execution of a rule of \A since
  $\PReset(u)$ holds while $u$ does not switch to status $C$
  (Requirements~\ref{RQ4} and \ref{RQ3}). Then, the next rule of \A
  $u$ will execute (if any) is $\ruleC$. However, if this latter case
  happens, $\ruleRF(u)$ and $\ruleC(u)$ are executed in different
  segments, by Lemma~\ref{lem:rfSeg}.
\end{proof}

Since a process can execute rules of {\tt I} only if its status is
$C$, we have the following corollary.

\begin{corollary}\label{cor:seg}
 The sequence of rules executed by a process $u$ in a segment of execution of {\tt I} $\circ$ \A  belongs to the following language:
 $$(\ruleC + \varepsilon)\ \wordsI\ (\ruleRB + \ruleR + \varepsilon)\ (\ruleRF + \varepsilon)$$
 where $\wordsI$ is any sequence of rules of {\tt I}.
\end{corollary}

From Remark~\ref{rem:nbseg} and Theorem~\ref{theo:seg}, follows.

\begin{corollary}
Any process $u$ executes at most $3n+3$ rules of \A in any execution  of {\tt I} $\circ$ \A. 
\end{corollary}

Let $S = \gamma_0 \cdots \gamma_j \cdots$ be a segment of execution of
{\tt I} $\circ$ \A.
Let $c_{\tt I}^S$ be the configuration of {\tt I} in which every
process $u$ has the local state $c_{\tt I}^S(u)$ defined below.
  \begin{enumerate}
 \item  $c_{\tt I}^S(u) = \gamma_{0|{\tt I}}(u)$, if $u$ never satisfies $\st_u =
   C$ in $S$,
   \item $c_{\tt I}^S(u) = \gamma_{i|{\tt I}}(u)$ where $\gamma_i$ is
  the first configuration such that $\st_u = C$ in $S$, otherwise.
\end{enumerate}
The following lemma is a useful tool to show the convergence of {\tt
  I} $\circ$ \A.

\begin{lemma}\label{lem:transfert}
  Let $S = \gamma_0 \cdots \gamma_j \cdots$ be a segment of execution
  of {\tt I} $\circ$ \A. For every process $u$, let $\wordsI(u)$ be
  the (maybe empty) sequence of rules of {\tt I} executed by $u$ in
  $S$.  There is a prefix of execution of {\tt I} starting in $c_{\tt
    I}^S$ that consists of the executions of $\wordsI(u)$, for every
  process $u$.
  \end{lemma}
\begin{proof}
  For every process $u$, for every $i \geq 0$, let
  $\prewordsI(\gamma_i,u)$ be the prefix of $\wordsI(u)$ executed by
  process $u$ in the prefix $\gamma_0 \cdots \gamma_i$ of $S$.  For
  every $i \geq 0$, let $\confI(\gamma_i)$ be the configuration of
  {\tt I} in which every process $u$ has the local state
  $\confI(\gamma_i)(u)$ defined below.
\begin{itemize}
\item $\confI(\gamma_i)(u)= c_{\tt I}^S(u)$, if $\prewordsI(\gamma_i,u)
  = \emptyset$,
\item $\confI(\gamma_i)(u) = s$ where $s$ is the state assigned to $u$
  by the execution of its last rule of $\prewordsI(\gamma_i,u)$ in the
  prefix $\gamma_0 \cdots \gamma_i$ of $S$, otherwise.
\end{itemize}
The lemma is immediate from the following induction: for every $i\geq
0$, there is a possible prefix of execution of {\tt I} that starts
from $c_{\tt I}^S$, ends in $\confI(\gamma_i)$, and consists of the
executions of $\prewordsI(\gamma_i,u)$, for every process $u$.
The base case ($i = 0$) is trivial. 
Assume the induction holds for
some $i\geq 0$ and consider the case $i+1$. 
For every process $u$, either $u$ does not execute any rule of {\tt I}
in $\gamma_{i} \mapsto \gamma_{i+1}$ and $\prewordsI(\gamma_{i+1},u) =
\prewordsI(\gamma_{i},u)$, or $u$ executes some rule $R_u$ of {\tt I}
in $\gamma_{i} \mapsto \gamma_{i+1}$ and $\prewordsI(\gamma_{i+1},u) =
\prewordsI(\gamma_{i},u)\cdot R_u$.  
If all processes satisfy the
former case, then, by induction hypothesis, we are done. 
Otherwise, by induction hypothesis, it is sufficient to show the
transition from $\confI(\gamma_i)(u)$ to $\confI(\gamma_{i+1})(u)$
consisting of the execution of $R_u$ by every process $u$ such that
$\prewordsI(\gamma_{i+1},u) = \prewordsI(\gamma_{i},u)\cdot R_u$ is a
possible step of {\tt I}.
To see this,
consider any process $u$ such that $\prewordsI(\gamma_{i+1},u) =
\prewordsI(\gamma_{i},u)\cdot R_u$. Since $u$ executes $R_u$ in
$\gamma_{i} \mapsto \gamma_{i+1}$, $\PClean(u)$ holds in $\gamma_{i}$
meaning that every member $v$ of $N[u]$ (in particular, $u$) satisfies
$\st_v = C$ in $\gamma_i$. Then, by Corollary \ref{cor:seg}, every $v$
(in particular, $u$) is in the state $\confI(\gamma_i)(u)$. Hence, as
$R_u$ is enabled in $\gamma_i$, $R_u$ is enabled in $\confI(\gamma_i)$
too, and we are done.
\end{proof}

\subsubsection{Round Complexity.}

Below, we use the notion of attractor, defined at the beginning of the section.

\begin{definition}[Attractors]$~$
\begin{itemize}
\item Let $\mathcal{P}_1$ a predicate over configurations of {\tt I}
  $\circ$ \A which is true if and only if $\neg\PRUp(u)$ holds, for
  every process $u$.

\item Let $\mathcal{P}_2$ a predicate over configurations of {\tt I}
  $\circ$ \A which is true if and only if (1) $\mathcal{P}_1$ holds
  and (2) $\neg\PB(u)$ holds, for every process $u$.

\item Let $\mathcal{P}_3$ a predicate over configurations of {\tt I}
  $\circ$ \A which is true if and only if (1) $\mathcal{P}_2$ holds
  and (2) $\status_u \neq RB$, for every process $u$.

\item $\mathcal{P}_4$ a predicate over configurations of {\tt I}
  $\circ$ \A which is true if and only if (1) $\mathcal{P}_3$ holds
  and (2) $\status_u \neq RF$, for every process $u$.

  In the following, we call {\em normal configuration} any
  configuration satisfying $\mathcal{P}_4$.
\end{itemize}
\end{definition}

\begin{lemma}
  \label{lem:att1} $\mathcal{P}_1$ is an attractor for {\tt I} $\circ$
  \A. Moreover, {\tt I} $\circ$ \A converges from $true$ to
  $\mathcal{P}_1$ within at most one round.
\end{lemma}
\begin{proof}
  First, for every process $u$, $\neg\PRUp(u)$ is closed by {\tt I}
  $\circ$ \A (Corollary \ref{cor:correct}). Consequently,
  $\mathcal{P}_1$ is closed by {\tt I} $\circ$ \A. Moreover, to show
  that {\tt I} $\circ$ \A converges from $true$ to $\mathcal{P}_1$ within at most
  one round, it is sufficient to show that any process $p$ satisfies
  $\neg\PRUp(u)$ during the first round of any execution of {\tt I}
  $\circ$ \A. This property is immediate from the following two claims.
\begin{description}
\item[Claim 1:] {\em If $\PRUp(u)$, then $u$ is enabled.}

{\em Proof of the claim:}  By definition of $\ruleR(u)$.
\item[Claim 2:] {\em If $\PRUp(u)$ holds in $\gamma$ and $u$ moves in the next step $\gamma \mapsto \gamma'$, then $\neg\PRUp(u)$ holds in $\gamma'$.}

  {\em Proof of the claim:} First, by Remark~\ref{rem:mutex},
  Lemma~\ref{lem:mutex:2}, and the guard of $\ruleR(u)$, $\ruleR(u)$
  is executed in $\gamma \mapsto \gamma'$. Then, immediately after
  $\ruleR(u)$, $\status_u = RB$ and $\PReset(u)$ holds (see
  Requirement~\ref{RQ5}), now $\status_u = RB$ and $\PReset(u)$
  implies $\neg\PRUp(u)$.
\end{description}

\end{proof}

\begin{lemma}
\label{lem:att2-Closure}
$\mathcal{P}_2$ is closed by {\tt I} $\circ$ \A.
\end{lemma}
\begin{proof} By Requirement~\ref{RQ1}, $\neg\PB(u)$ is closed by {\tt
    I}, for every process $u$. So, by Lemma~\ref{lem:att1} and
  Remark~\ref{rem:mutex}, it is sufficient to show that for every step
  $\gamma \mapsto \gamma'$ of  {\tt I} $\circ$ \A such that $\mathcal{P}_2$ holds in $\gamma$, for every process $u$, if $u$ executes a rule of \A in $\gamma \mapsto \gamma'$, then $\neg\PB(u)$ still holds in $\gamma'$.

So, assume any such step $\gamma \mapsto \gamma'$ and any process $u$. 
\begin{itemize}
\item If $\status_u = C$ in $\gamma$, then $\forall v \in \N(u)$, $\status_v \neq RB$ in $\gamma$, since $\gamma$ satisfies $\mathcal{P}_2$. Now, no rule $\ruleRB$ or $\ruleR$ can be executed in $\gamma \mapsto \gamma'$ since $\gamma$ satisfies $\mathcal{P}_2$. So, $\forall v \in \N(u)$, $\status_v \neq RB$ in $\gamma'$, and consequently $\neg\PB(u)$ still holds in $\gamma'$.
\item If $\status_u \neq C$ in $\gamma'$, then  $\neg\PB(u)$ holds in $\gamma'$.
\item Assume now that $\status_u \neq C$ in $\gamma$ and  $\status_u = C$ in $\gamma'$. Then, $u$ necessarily executes $\ruleC$ in $\gamma \mapsto \gamma'$. In this case,  $\forall v \in \N(u)$, $\status_v \neq RB$ in $\gamma$. Now, no rule $\ruleRB$  or $\ruleR$ can be executed in $\gamma \mapsto \gamma'$ since $\gamma$ satisfies $\mathcal{P}_2$. So, $\forall v \in \N(u)$, $\status_v \neq RB$ in $\gamma'$, and consequently $\neg\PB(u)$ still holds in $\gamma'$.
\end{itemize}
Hence, in all cases,  $\neg\PB(u)$ still holds in $\gamma'$, and we are done.
\end{proof}

\begin{lemma}
  \label{lem:att2} {\tt I} $\circ$ \A converges from $\mathcal{P}_1$
  to $\mathcal{P}_2$ within at most $n-1$ rounds.
\end{lemma}
\begin{proof}
  Let $u$ be any process of status $RB$. Then, $u$ belongs to at least
  one reset branch (Remark~\ref{rem:unless}). 
  Let $md(u)$ be the maximum depth of $u$ in a reset branch it belongs
  to. Then,
%
$md(u) < n$, by Lemma~\ref{lem:trace}.

Consider now any execution $e = \gamma_0 \cdots \gamma_i \cdots$ of {\tt I} $\circ$ \A such that $\gamma_0$ satisfies
  $\mathcal{P}_1$. Remark first that from $\gamma_0$,
  $\ruleR(v)$ is disabled forever, for every process $v$, since
  $\mathcal{P}_1$ is closed by {\tt I} $\circ$ \A
  (Lemma~\ref{lem:att1}).

\begin{description}
\item[Claim 1:] {\em If some process $u$ satisfies $\status_{u} = RB$ in
  some configuration $\gamma_i$ ($i \geq 0$), then from $\gamma_i$,
  while $\status_{u} = RB$, $md(u)$ cannot decrease.}

  {\em Proof of the claim:} Consider any $\gamma_i \mapsto
  \gamma_{i+1}$ where $\status_{u} = RB$ both in $\gamma_i$ and
  $\gamma_{i+1}$. This in particular means that $u$ does not move in
  $\gamma_i \mapsto \gamma_{i+1}$. Let $u_1, \ldots, u_k=u$ be a
  reset branch in $\gamma_i$ such that $k = md(u)$. $\forall x \in
  \{1,\ldots,k-1\}$, $u_x$ has a neighbor ($u_{x+1}$) such that
  $\status_{x+1} = RB \wedge \dist_{x+1} > \dist_u$ in $\gamma_i$, by
  Lemma~\ref{lem:trace} and the definition of a reset branch. Hence,
  every $u_x$ is disabled in $\gamma_i$. Consequently, $u$ is still at
  depth at least $k$ in a reset branch defined in $\gamma_{i+1}$, and
  we are done.

\item[Claim 2:] {\em For every process $u$ that executes $\ruleRB(u)$ in
  some step $\gamma_i \mapsto \gamma_{i+1}$ of the $j$th round of $e$, we have $\status_{u} = RB \wedge
  md(u) \geq j$ in $\gamma_{i+1}$.}

  {\em Proof of the claim:}  We proceed by induction. Assume a process $u$ executes
  $\ruleRB(u)$ in some step $\gamma_i \mapsto \gamma_{i+1}$ of the
  first round of $e$. In $\gamma_i$, there is some neighbor
  $v$ of $u$ such that $\status_v = RB$. Since $\status_u = C$ in
  $\gamma_i$, $v$ is disabled in $\gamma_i$. Consequently,
  $RParent(v,u)$ holds in $\gamma_{i+1}$, and so $\status_{u} = RB
  \wedge md(u) \geq 1$ holds in $\gamma_{i+1}$. 
%
Hence, the claim holds for $j = 1$.
 
  Assume now that the claim holds in all of the $j$ first rounds of $e$, with $j \geq 1$.

  Assume, by the contradiction, that some process $u$ executes
  $\ruleRB$ in a step $\gamma_i \mapsto \gamma_{i+1}$ of the
  $(j+1)^{th}$ round of $e$, and does not satisfy $\status_{u} = RB
  \wedge md(u) \geq j+1$ in $\gamma_{i+1}$. Then, by definition of
  $\ruleRB(u)$, $\status_{u} = C$ in $\gamma_i$ and $\status_{u} = RB
  \wedge md(u) < j+1$ in $\gamma_{i+1}$.  Let $x$ be the value of
  $md(u)$ in $\gamma_{i+1}$. We have $x < j+1$.  Without the loss of
  generality, assume that no process satisfies this condition before
  $u$ in the $(j+1)^{th}$ round of $e$ and any process $v$ that
  fulfills this condition in the same step as $u$ satisfies
  $\status_{v} = RB \wedge x \leq md(v) < j+1$ in $\gamma_{i+1}$.
  Then, by definition $md(u)$ and Lemma~\ref{lem:trace}, there is a
  neighbor $v$ of $u$ such $\status_v = RB$ and $md(v) = x-1 < j$ in
  $\gamma_{i+1}$. Moreover, by definition of $u$ and Claim 1,
  $\status_v = RB$ and $md(v) \leq x-1 < j$ since (at least) the first
  configuration of the $(j+1)$th round of $e$, which is also the last
  configuration of the $j$th round of $e$. So, by induction
  hypothesis, $\status_v = RB$ and $md(v) \leq x-1 < j$ since (at
  least) the end of the $(x-1)$th round of $e$.
  If $\status_u \neq C$ in the last configuration of the $(x-1)$th
  round of $e$, then $\status_u \neq C$ continuously until $\gamma_i$
  (included) since meanwhile $\ruleC(u)$ is disabled because
  $\status_v = RB$. Hence, $u$ cannot execute $\ruleRB(u)$ in
  $\gamma_i \mapsto \gamma_{i+1}$, a contradiction. Assume otherwise
  that $\status_u = C$ in the last configuration of the $(x-1)$th
  round of $e$. Then, $u$ necessarily executes $\ruleRB$ during the
  $x$th round of $e$, but not in the $(j+1)^{th}$ round of $e$ since
  $\status_v = RB$ continuously until $\gamma_i$ (included), indeed,
  after the execution of $\ruleRB(u)$ in the $x$th round, the two next
  rules executed by $u$ (if any) are necessarily $\ruleRF$ followed by
  $\ruleC$, but $\ruleC(u)$ is disabled while $\status_v = RB$. Hence,
  $\ruleRB(u)$ is not executed in $\gamma_i \mapsto \gamma_{i+1}$, a
  contradiction.
\end{description}
By Claim 2, no process executes $\ruleRB$ during the $n$th round of
$e$. Now, along $e$, we have:
\begin{itemize}
\item If $\PB(u)$ holds, then $u$ is enabled (see $\ruleRB(u)$), and
\item If $\PB(u)$ holds in $\gamma_i$ (with $i \geq 0$) and $u$ moves
  in the next step $\gamma_i \mapsto \gamma_{i+1}$, then $\neg\PB(u)$
  holds in $\gamma_{i+1}$.

  Indeed, $u$ necessarily executes $\ruleRB(u)$ in $\gamma_i \mapsto
  \gamma_{i+1}$ (Remark~\ref{rem:mutex} and Lemma ~\ref{lem:mutex:2})
  and, consequently, $\status_u = RB$ in $\gamma_{i+1}$, which implies
  that $\neg\PB(u)$ holds in $\gamma_{i+1}$.
\end{itemize} 
Hence, we can conclude that the last configuration of the $(n-1)$th
round of $e$ satisfies $\mathcal{P}_2$, and we are done.
\end{proof}

\begin{lemma}
\label{lem:att3}
$\mathcal{P}_3$ is closed by {\tt I} $\circ$ \A.
Moreover, {\tt I} $\circ$ \A converges from $\mathcal{P}_2$
  to $\mathcal{P}_3$ within at most $n$ rounds.
\end{lemma}
\begin{proof}
  By Requirement \ref{RQ1}, no rule of {\tt I} can set the status of a
  process to $RB$. Then, let $\gamma$ be a configuration of {\tt I}
  $\circ$ \A such that $\mathcal{P}_3(\gamma)$ holds. Since
  $\mathcal{P}_1(\gamma)$ and $\mathcal{P}_2(\gamma)$ also holds, no
  rule $\ruleR$ or $\ruleRB$ is enabled in $\gamma$. Hence, after any
  step from $\gamma$, there is still no process of status $RB$ and we
  can conclude that $\mathcal{P}_3$ is closed by {\tt I} $\circ$ \A
  since we already know that $\mathcal{P}_2$ is closed by {\tt I}
  $\circ$ \A (Lemma~\ref{lem:att2-Closure}).

  Let $\gamma$ be any configuration satisfying $\mathcal{P}_2$ but not
  $\mathcal{P}_3$. To show the convergence from from $\mathcal{P}_2$
  to $\mathcal{P}_3$ within at most $n$ rounds, it is sufficient to
  show that at least one process $u$ switches from $\status_u = RB$ to
  $\status_u \neq RB$ within the next round from $\gamma$, since we
  already know that once $\status_u \neq RB$ after $\gamma$,
  $\status_u \neq RB$ holds forever (recall that all configurations
  reached from $\gamma$ satisfies $\mathcal{P}_2$; 
  see Lemma~\ref{lem:att2-Closure}).

  Let $mu$ be a process of status $RB$ with a maximum distance value
  in $\gamma$. Since $\neg\PRUp(mu) \wedge \neg\PB(mu) \wedge
  \status_{mu} \neq C$ holds, $\PReset(mu)$ holds in $\gamma$.  Let
  $v$ be any neighbor of $mu$.  Since $\neg\PB(v) \wedge\status_{mu} =
  RB$ holds, $\status_v \neq C$ in $\gamma$. Again, since
  $\neg\PRUp(v) \wedge \neg\PB(v) \wedge \status_v \neq C$ holds,
  $\PReset(v)$ holds in $\gamma$.  According to the definition of
  $mu$, we have $(\status_v = RB \wedge \dist_v \leq \dist_{mu}) \Or
  \status_v = RF$.  We can conclude that along any execution from
  $\gamma$, $\ruleRF(mu)$ is enabled until $mu$ executes this rule.
  In this case, $\ruleRF(mu)$ will be executed in the next move of
  $mu$ by Remark~\ref{rem:mutex} and Lemma~\ref{lem:mutex:2}.  So, during
  the next round from $\gamma$, $\ruleRF(mu)$ is executed, {\em i.e.},
  $\status_{mu}$ is set to $RF$, and we are done.
\end{proof}

Let $\gamma$ by any configuration of {\tt I} $\circ$ \A. We denote by $\gamma_{|\A}$ the projection of $\gamma$ over variables of \A. By definition, $\gamma_{|\A}$ is a configuration of \A.

\begin{lemma}\label{lem:a4:term}
For every configuration $\gamma$ of 
{\tt I} $\circ$ \A, $\gamma \in \mathcal{P}_4$ ({\em i.e.}, $\gamma$ is a normal configuration) if and only if $\gamma_{|\A}$ is a terminal configuration of \A.
\end{lemma}
\begin{proof}
  Let $\gamma$ be a configuration of {\tt I} $\circ$ \A.  By
  definition of $\mathcal{P}_4$, if $\gamma \in \mathcal{P}_4$, then
  $\gamma_{|\A}$ is a terminal configuration of \A. Then, if
  $\gamma_{|\A}$ is a terminal configuration of \A, then $\gamma \in
  \mathcal{P}_4$ by Lemmas~\ref{lem:nonPRabc}, \ref{lem:nonRB}, and
  \ref{lem:nonRF}.
\end{proof}

\begin{lemma}
\label{lem:att4}
$\mathcal{P}_4$ is closed by {\tt I} $\circ$ \A.
Moreover, {\tt I} $\circ$ \A converges from $\mathcal{P}_3$
  to $\mathcal{P}_4$ within at most $n$ rounds.
\end{lemma}
\begin{proof}
  By Requirement \ref{RQ1}, no rule of {\tt I} can set the status of a
  process to $RF$. So, by Lemma~\ref{lem:a4:term}, we can conclude
  that $\mathcal{P}_4$ is closed by {\tt I} $\circ$ \A.

 Let $\gamma$ be any configuration satisfying $\mathcal{P}_3$ but not
  $\mathcal{P}_4$. To show the convergence from from $\mathcal{P}_3$
  to $\mathcal{P}_4$ within at most $n$ rounds, it is sufficient to
  show that at least one process $u$ switches from $\status_u = RF$ to
  $\status_u \neq RF$ within the next round from $\gamma$, since we
  already know that once $\status_u \neq RF$ after $\gamma$,
  $\status_u \neq RF$ holds forever (recall that all configurations
  reached from $\gamma$ satisfies $\mathcal{P}_3$ by Lemma~\ref{lem:att3}, and only process of status $RB$ may switch to status $RF$).

  Let $mu$ be a process of status $RF$ with a minimum distance value
  in $\gamma$. Since $\neg\PRUp(mu) \wedge \neg\PB(mu) \wedge
  \status_{mu} \neq C$ holds, $\PReset(mu)$ holds in $\gamma$.  Let
  $v$ be any neighbor of $mu$.  By definition of $\mathcal{P}_3$,
  $\status_v \neq RB$ in $\gamma$.  Moreover, since $\neg\PRUp(v)
  \wedge \neg\PB(v)$ and $v$ has a neighbor of status $RF$ ($mu$),
  $\PReset(v)$ holds in $\gamma$.  According to the definition of
  $mu$, we have $(\status_v = RF \wedge \dist_v \geq \dist_{mu}) \Or
  \status_v = C$.  We can conclude that along any execution from
  $\gamma$, $\ruleC(mu)$ is enabled until $mu$ executes this rule.  In
  this case, $\ruleC(mu)$ will be executed in the next move of $mu$ by
  Remark~\ref{rem:mutex} and Lemma~\ref{lem:mutex:2}.  So, during the
  next round from $\gamma$, $\ruleC(mu)$ is executed, {\em i.e.},
  $\status_{mu}$ is set to $C$, and we are done.
\end{proof}

By Lemmas~\ref{lem:att1}-\ref{lem:att4} and Theorem~\ref{theo:termA}, follows.

\begin{corollary}
\label{cor:round}
$\mathcal{A}_4$ is an attractor for {\tt I} $\circ$
  \A. Moreover, {\tt I} $\circ$ \A converges from $true$ to
  $\mathcal{P}_4$ within at most $3n$ rounds. 
For every configuration $\gamma$ of {\tt I} $\circ$ \A, $\gamma$ satisfies  $\mathcal{A}_4$ ({\em i.e.} $\gamma$ is a normal configuration) if and only if  $\PClean(u)
\And \PICorrect(u)$ holds in $\gamma$, for every process $u$.
\end{corollary}

\section{Asynchronous Unison}\label{sect:unison}

\subsection{The Problem}

We now consider the problem of {\em asynchronous unison} (introduced
in ~\cite{CFG92c}), simply referred to as unison in the
following. This problem is a clock synchronization problem: each
process $u$ holds a variable (usually an integer variable) called 
\emph{clock}, here noted $\val_u$. Then, the problem is specified as
follows:
\begin{itemize}
\item Each process should increment its clock infinitely often.
  (liveness)
  \item The difference between clocks of every two neighbors should be
    at most one increment at each instant. (safety)
  \end{itemize}
Notice that we consider here periodic clocks, {\em i.e.}, the clock
incrementation is modulo a so-called \emph{period}, here noted
$\ParU$.

\subsection{Related Work}

The first self-stabilizing asynchronous unison for general connected
graphs has been proposed by Couvreur {\em et al.}~\cite{CFG92c}. It is
written in the locally shared memory model with composite atomicity
assuming a central unfair daemon and a period $\ParU > n^2$. No
complexity analysis was given.
Another solution which stabilizes in $O(n)$ rounds has been proposed
by Boulinier {\em et al.} in \cite{BPV04c}. This solution is also
written in the locally shared memory model with composite atomicity,
however it assumes a distributed unfair daemon.  In this solution, the
period $\ParU$ should satisfy $\ParU > C_G$ and another parameter
$\alpha$ should satisfy $\alpha\ge T_G-2$. $C_G$ is the {\em
  cyclomatic characteristic\/} of the network and $T_G$ is the length
of the longest chordless cycle.  Boulinier also proposed in his PhD
thesis~\cite{B07t} a parametric solution which generalizes both the
solutions of \cite{CFG92c} and \cite{BPV04c}. In particular, the study
of this parametric algorithm reveals that the solution of Couvreur
{\em et al.}~\cite{CFG92c} still works assuming a distributed unfair
daemon and has a stabilization time in $O(D.n)$ rounds, where $D$ is
the network diameter.

\subsection{Contribution}

We first propose a distributed algorithm, called \DU. Starting from a
pre-defined configuration, \DU implements the unison problem in
anonymous networks, providing that the period $\ParU$ satisfies $\ParU
> n$.  \DU is not self-stabilizing, however we show that the composite
algorithm \DU $\circ$ \A is actually an efficient self-stabilizing
unison algorithm.  Indeed, its stabilization times in round matches
the one of the best existing solution~\cite{BPV04c}. Moreover, it
achieves a better stabilization time in moves, since it stabilizes in
$O(D.n^2)$ moves, while the algorithm in~\cite{BPV04c} stabilizes in
$O(D.n^3+\alpha.n^2)$ moves; as shown in~\cite{DP12}.

\subsection{Algorithm \DU}

\paragraph{Overview.}
We consider here anonymous (bidirectional) networks of arbitrary connected
topology. Moreover, every process has the period $\ParU$ as
input. $\ParU$ is required to be (strictly) greater than $n$, the
number of processes. 
The formal code of Algorithm \DU is given in Algorithm~\ref{alg:DU}.
Informally, each process maintains a single variable, its clock
$\val_u$, using a single rule $\ruleUA(u)$.

\begin{algorithm}
\small
  $~$ \\[0.3cm]
\textbf{Inputs:} \\[0.1cm]
\begin{tabular}{llll}
  $\bullet$ & $\st_u \in \{ C, RB, RF \}$ & : & variable of \A\\
  $\bullet$ & $\PClean(u)$ & : & predicate of \A\\
  $\bullet$ & $\ParU$ & : & a constant from the system satisfying $\ParU > n$
\end{tabular}
\\[0.3cm]
\textbf{Variables:} \\[0.1cm]
\begin{tabular}{llll}
$\bullet$ & $\val_u \in \mathds{N}$ & : & the clock of $u$
\end{tabular}
\\[0.3cm]
\textbf{Predicates:} \\[0.1cm]
\begin{tabular}{llll}
  $\bullet$ & $\PAgree(u,v)$  & $\equiv$ & $\val_v \in \{(\val_u -1) \bmod \ParU, \val_u, (\val_u +1) \bmod \ParU\}$  \\
  $\bullet$ & $\PICorrect(u)$ & $\equiv$ & $(\forall v \in \N(u), \PAgree(u,v))$ \\
  $\bullet$ & $\PReset(u)$    & $\equiv$ & $\val_u=0$\\
  $\bullet$ & $\PUp(u)$ & $\equiv$ & $(\forall v \in \N(u), \val_v \in \{\val_u, (\val_u +1) \bmod \ParU\})$
\end{tabular}  
\\[0.3cm]
\textbf{Macros:}\\[0.1cm]
\begin{tabular}{llll}
$\bullet$ & $\reset(u)$ & : & $\val_u:= 0$;
\end{tabular}
\\[0.3cm]
{\textbf{Rules:}} \\[0.1cm]
\begin{tabular}{lllll}
$\ruleUA(u)$ & : &  $\PClean(u) \And \PUp(u)$ & 
$\to$ & 
$\val_u := (\val_u+1)\bmod \ParU$; 
 \end{tabular}  
\\[0.3cm]
\caption{Algorithm \DU, code for every process $u$}
\label{alg:DU}
\end{algorithm}

In the following, we assume that the system is initially in the
configuration $\gamma_{init}$ where every process $u$ satisfies
$\val_u = 0 \wedge  \st_u = C$. Basically, starting from $\gamma_{init}$, a process $u$
can increment its clock $\val_u$ modulo $\ParU$ (using rule
$\ruleUA(u)$) if it is on time or one increment late with each of its
neighbors; see predicate $\PUp(u)$.

\paragraph{Correctness.}

Below, we focus on configurations of \DU satisfying $\PICorrect(u)
\wedge \PClean(u)$, for every process $u$. Indeed, $\gamma_{init}$ belongs
to this class of configurations. Moreover, a configuration of \DU $\circ$ \A
is normal if and only if $\PClean(u) \And \PICorrect(u)$ holds for
every process $u$. Hence, the properties we exhibit now will be, in
particular, satisfied at the completion of \A.

Consider any two neighboring processes $u$ and $v$ such that
$\PAgree(u,v)$ in some configuration $\gamma$, {\em i.e.}, $\val_v \in
\{(\val_u -1) \bmod \ParU, \val_u, (\val_u +1) \bmod \ParU\}$. Without
the loss of generality, assume that $\val_v \in \{\val_u, (\val_u +1)
\bmod \ParU\}$ (otherwise switch the role of $u$ and $v$). Let $\gamma
\mapsto \gamma'$ be the next step.  If $\val_v = \val_u$ in $\gamma$,
then $\val_v \in \{(\val_u -1) \bmod \ParU, \val_u, (\val_u +1) \bmod
\ParU\}$ in $\gamma'$ since each clock increments at most once per
step, {\em i.e.}, $\PAgree(u,v)$ still holds in $\gamma'$. Otherwise,
$\val_v = (\val_u +1) \bmod \ParU\}$ in $\gamma$, and so $v$ is
disabled and only $u$ may move. If $u$ does not move, then $\val_v =
(\val_u +1) \bmod \ParU\}$ still holds in $\gamma'$, otherwise $\val_v
= \val_u$ in $\gamma'$. Hence, in both cases, $\PAgree(u,v)$ still
holds in $\gamma'$. Hence, follows.

\begin{lemma}\label{unison:agree}
$\PICorrect(u)$ is closed by \DU, for every process $u$.
\end{lemma}

Since Algorithm \DU does not modify any variable from Algorithm \A, we have

\begin{remark}\label{lem:PClean:clos2}
$\PClean(u)$ is closed by \DU, for every process $u$.
\end{remark}

\begin{corollary}\label{unison:safety} 
  $\PICorrect(u) \wedge \PClean(u)$ is closed by \DU, for every process $u$.
  \end{corollary}

\begin{corollary}\label{u:safe}
Any execution of \DU, that starts from a configuration where
$\PICorrect(u) \wedge \PClean(u)$ holds for every process $u$,
satisfies the safety of the unison problem.
  \end{corollary}

\begin{lemma}\label{unison:noDL}
Any configuration where
$\PICorrect(u) \wedge \PClean(u)$ holds for every process $u$ is not terminal.
\end{lemma}
\begin{proof}
Assume, by the contradiction, a terminal configuration $\gamma$  where
$\PICorrect(u) \wedge \PClean(u)$  for every process $u$. Then, every process $u$ has at least one neighbor $v$ such that $\val_v = (\val_u -1) \bmod \ParU$. Since the number of processes is finite, in $\gamma$ there exist elementary cycles $u_1, \ldots, u_x$ such that
\begin{enumerate}
\item for every $i \in \{1, \ldots, x-1\}$, $u_i$ and $u_{i+1}$ are neighbors, and $\val_{u_i} =
  (\val_{u_{i+1}} -1) \bmod \ParU$; and \label{cycle:case1}
  \item $u_1$ and $u_x$ are neighbors and $\val_{u_x} =
    (\val_{u_1} -1) \bmod \ParU$. \label{cycle:case2}
\end{enumerate}
By transitivity, Case \ref{cycle:case1} implies that $\val_{u_1} =
(\val_{u_x} -(x-1)) \bmod \ParU$. So, from Case 2, we obtain $\val_{u_x} =
    (\val_{u_x} -x) \bmod \ParU$. Now, by definition $x \leq n$ and $\ParU > n$ so $\val_{u_x} \neq 
    (\val_{u_x} -x) \bmod \ParU$, a contradiction. Hence, $\gamma$ is not terminal.
\end{proof}

\begin{lemma}\label{u:fair}
Any execution of \DU, that starts from a configuration where
$\PICorrect(u) \wedge \PClean(u)$ holds for every process $u$,
satisfies the liveness of the unison problem.
\end{lemma}
\begin{proof}
  Let $e$ be 
any execution of \DU that starts from a configuration where
$\PICorrect(u) \wedge \PClean(u)$ holds for every process $u$. Assume, by the contradiction, that $e$ does not satisfy the liveness of unison. Then, $e$ contains a configuration $\gamma$ from which some processes (at least one) never more executes $\ruleUA$. Let $F$ be the non-empty subset of processes that no more move from $\gamma$. Let $I = V \setminus F$. By Lemma~\ref{unison:noDL}, $I$ is not empty too. Now, since the network is connected, there are two processes $u$ and $v$ such that $u \in I$ and $v \in F$. Now, after at most 3 increments of $u$ from $\gamma$, $\PAgree(u,v)$ no more holds, contradicting Lemma~\ref{unison:agree}.
\end{proof}

Consider now any execution $e$ of \DU starting from
$\gamma_{init}$. In $\gamma_{init}$, we have $\PClean(u) \wedge
\PICorrect(u)$ for every process $u$. Hence, by Corollary~\ref{u:safe} and Lemma~\ref{u:fair}, follows.

\begin{theorem}\label{theo:u:correct}
  \DU is distributed (non self-stabilizing) unison.
\end{theorem}

\paragraph{Properties of \DU.}

Consider any execution $e$ of \DU starting from a configuration
$\gamma$ which does not satisfy $\PClean(u) \wedge \PICorrect(u)$ for
every process $u$.  Then, there exists at least one process $u$
satisfying $\neg \PClean(u) \vee \neg \PICorrect(u)$ in $\gamma$, and
$u$ is disabled forever in $e$. Indeed, if $\neg \PClean(u)$,
then $\neg \PClean(u)$ holds forever since \DU does not write into
\A's variables. If $ \neg \PICorrect(u)$ holds, then there is a
neighbor $v$ such that $\neg \PAgree(u,v)$ holds, both $u$ and $v$ are
disabled, hence so $\neg \PAgree(u,v)$ holds forever, which implies
that $\neg \PUp(u)$ forever.
Now, since $u$ is disabled forever, each neighbor of $u$ moves at most
three times in $e$. Inductively, every node at distance $d$ from $u$
moves at most $3d$ times. Overall, we obtain the following lemma.

\begin{lemma}\label{lem:move:unison}
In any execution of \DU starting from a configuration which does not
satisfy $\PClean(u) \wedge \PICorrect(u)$ for every process $u$, each
process moves at most $3D$ times, where $D$ is the network diameter.
\end{lemma}

\subsection{Algorithm \DU $\circ$ \A}

\paragraph{Requirements.}

To show the self-stabilization of \DU $\circ$ \A, we should first
establish that \DU meets the requirements~\ref{RQ1} to~\ref{RQ6},
given in Subsection~\ref{sect:require}.

Requirement~\ref{RQ2} is satisfied since $\PICorrect(u)$ does not
involve any variable of \A and is closed by \DU
(Lemma~\ref{unison:agree}).  All other requirements directly follow
from the code of \DU.

\paragraph{Self-stabilization and Move Complexity.}

We define the legitimate configurations of \DU $\circ$ \A as the set
of configurations satisfying $\PClean(u) \wedge \PICorrect(u)$ for
every process $u$. This set actually corresponds to the set of normal
configurations (see Corollary~\ref{cor:round}, page
\pageref{cor:round}) and 
is closed by Algorithm \DU $\circ$
\A, by Remark~\ref{rem:mutex} (page \pageref{rem:mutex}),
Theorem~\ref{theo:termA} (page \pageref{theo:termA}), and
Corollary~\ref{unison:safety}. Then, from any normal configuration,
the specification of the unison holds, by Corollary~\ref{u:safe} and
Lemma~\ref{u:fair}.  So, it remains to show the convergence.

Let $u$ be a process. Let $e$ be an execution of \DU $\circ$ \A.
By Corollary \ref{cor:seg} (page \pageref{cor:seg}), the sequence of
rules executed by a process $u$ in a segment of $e$ belongs to the
following language:
$$(\ruleC + \varepsilon)\ \wordsI\ (\ruleRB + \ruleR +
\varepsilon)\ (\ruleRF + \varepsilon)$$ where $\wordsI$ be any
sequence of rules of \DU.

Let assume that $e$ contains $s$ segments. Recall that $s \leq n+1$;
see Remark~\ref{rem:nbseg}, page \pageref{rem:nbseg}.  Let call {\em
  regular} segment any segment that starts in a configuration
containing at least one abnormal alive root. A regular segment
contains no normal configuration, so, by Lemma~\ref{lem:transfert}
(page \pageref{lem:transfert}) and Lemma~\ref{lem:move:unison}, the
sequence $\wordsI$ of $u$ is bounded by $3D$ in $S$.  Thus, $u$
executes at most $3D+3$ moves in $S$ and, overall a regular segment
contains at most $(3D+3).n$ moves and necessarily ends by a step where
the number of abnormal alive root decreases. Hence, all $s-1$ first
segments are regular and the last one is not.
Overall, the last segment $S_{last}$ starts after at most
$(3D+3).n.(s-1)$ moves.  $S_{last}$ contains no abnormal alive root
and so, the sequence of rules executed by $u$ in $S_{last}$ belongs to
the following language: $(\ruleC +
\varepsilon)\ \wordsI$.\footnote{Otherwise, $\st_u = RB$ in some
  configuration of $S_{last}$, and that configuration contains an
  abnormal root, by Lemma~\ref{lem:trace} (page \pageref{lem:trace}), a
  contradiction.} If the initial configuration of $S_{last}$ contains
no process of status $RF$, then it is a normal configuration and so $s
= 1$, {\em i.e.}, $e$ is initially in a normal configuration.
Otherwise, let $v$ be a process such that $\st_v = RF$ in the initial
configuration of $S_{last}$ and no other process executes $\ruleC$
later than $v$.  Following the same reasoning as in
Lemma~\ref{lem:move:unison}, while $v$ does not execute $\ruleC$, each
process other than $v$ can execute at most $3D$ rules of \DU and one
$\ruleC$. Hence, there are at most $(3D+1).(n-1)+1$ moves in
$S_{last}$ before the system reaches a normal configuration.

 Since, in the worst case $s = n+1$, overall $e$ reaches a normal
 configuration in at most $(3D+3).n^2+(3D+1).(n-1)+1$ moves, and we
 have the following theorem.

\begin{theorem} \DU $\circ$ \A is
  self-stabilizing for the unison problem. Its stabilization time is
  in $O(D.n^2)$ moves.
\end{theorem}

\paragraph{Round Complexity.}

By Corollary \ref{cor:round} (page~\pageref{cor:round}), follows.

\begin{theorem}
The
stabilization time of \DU $\circ$ \A is  at most $3n$ rounds.
\end{theorem}



\section{$(f,g)$-alliance}\label{sect:alliance}

\subsection{The Problem}

The $(f,g)$-alliance problem has been defined by Dourado {\em et
  al.}~\cite{DouradoPRS11}.
Given a graph $G =(V,E)$, and two non-negative integer-valued
functions on nodes $f$ and $g$, a subset of nodes $A \subseteq V$ is
an {\em $(f,g)$-alliance} of $G$ if and only if every node $u\notin A$
has at least $f(u)$ neighbors in $A$, and every node $v\in A$ has at
least $g(v)$ neighbors in $A$.
The $(f,g)$-alliance problem is the problem of finding a subset of
processes forming an {\em $(f,g)$-alliance} of the network.
The {\em $(f,g)$-alliance} problem is a generalization of several
problems that are of interest in distributed computing. Indeed,
consider any subset $S$ of processes/nodes:
\begin{enumerate}
\item $S$ is a  domination set~\cite{berge2001theory} if and
  only if $S$ is a $(1,0)$-alliance;

\item more generally, $S$ is a $k$-domination
  set~\cite{berge2001theory} if and only if $S$ is a 
  $(k,0)$-alliance;

\item $S$ is a  $k$-tuple dominating set~\cite{LiaoC03} if
  and only if $S$ is a $(k,k-1)$-alliance;

\item $S$ is a  global offensive
  alliance~\cite{Sigarreta2009219} if and only if $S$ is a 
  $(f,0)$-alliance, where $f(u) = \lceil \frac{\delta_u+1}{2} \rceil$
  for all $u$;

\item $S$ is a  global defensive alliance~\cite{SigarretaR06}
  if and only if $S$ is a  $(1,g)$-alliance, where $g(u) =
  \lceil \frac{\delta_u+1}{2} \rceil$ for all $u$;

\item $S$ is a  global powerful alliance~\cite{YahiaouiBHK13}
  if and only if $S$ is a  $(f,g)$-alliance, such that
$f(u) = \lceil \frac{\delta_u+1}{2} \rceil$ and
  $g(u) = \lceil \frac{\delta_u}{2} \rceil$ for all $u$.
\end{enumerate}
We remark that $(f,g)$-alliances have applications in the fields of
population protocols~\cite{AngluinAER07} and server allocation in
computer networks~\cite{GuptaMOR05}.

Ideally, we would like to find a {\em minimum\/} $(f,g)$-alliance,
namely an $(f,g)$-alliance of the smallest possible cardinality.
However, this problem is $\mathcal{NP}$-hard, since the $(1,0)$-alliance ({\em i.e.}, the domination set problem) is known to be $\mathcal{NP}$-hard~\cite{GareyJ79}. 
We can instead consider the problem of finding a {\em minimal\/}
 $(f,g)$-alliance. 
An {\em $(f,g)$-alliance} is {\em minimal} if no proper
subset of $A$ is an {\em $(f,g)$-alliance}.
Another variant is the {\em $1$-minimal} {\em $(f,g)$-alliance}.
$A$ is a {\em $1$-minimal} {\em $(f,g)$-alliance} if
deletion of just one member of $A$ causes $A$ to be no more an {\em
  $(f,g)$-alliance}, {\em i.e.}, $A$ is an {\em $(f,g)$-alliance} but
$\forall u \in A$, $A \setminus \{u\}$ is not an {\em
  $(f,g)$-alliance}.  Surprisingly, a {\em $1$-minimal} {\em
  $(f,g)$-alliance} is not necessarily a {\em minimal} {\em
  $(f,g)$-alliance} \cite{DouradoPRS11}.
 However, we
have the following property:

\begin{property}[Dourado {\em et al.}~\cite{DouradoPRS11}]
  \label{prop:fga}
  Given two non-negative integer-valued functions $f$ and $g$ on nodes
  \begin{enumerate}
  \item Every minimal {\em $(f,g)$-alliance} is a $1$-minimal {\em
      $(f,g)$-alliance}, and
  \item if $f(u) \geq g(u)$ for every process $u$, then every $1$-minimal {\em $(f,g)$-alliance} is a
    minimal {\em $(f,g)$-alliance}.
  \end{enumerate}
\end{property}

\subsection{Contribution}

We first propose a distributed algorithm called \DA. Starting
from a pre-defined configuration, \DA computes a $1$-minimal
$(f,g)$-alliance in any identified network where $\delta_u \geq
max(f(u), g(u))$, for every process $u$.  Notice that this latter
assumption ensures the existence of a solution.  \DA is not
self-stabilizing, however we show that the composite algorithm \DA $\circ$ \A is actually an efficient self-stabilizing
$1$-minimal $(f,g)$-alliance algorithm.

\subsection{Related Work}

Recall that the {\em $(f,g)$-alliance} problem has been introduced by
Dourado {\em et al.}~\cite{DouradoPRS11}.  In that paper, the authors
give several distributed algorithms for that problem and its variants,
but none of them is self-stabilizing.

In~\cite{CDDLR15j}, Carrier {\em et al.} proposes a silent
self-stabilizing algorithm that computes a minimal {\em
  $(f,g)$-alliance} in an asynchronous network with unique node IDs,
assuming that every node $u$ has a degree at least $g(u)$ and
satisfies $f(u) \geq g(u)$.  Their algorithm is also {\em safely
  converging} in the sense that starting from any configuration, it
first converges to a (not necessarily minimal) {\em $(f,g)$-alliance}
in at most four rounds, and then continues to converge to a minimal
one in at most $5n + 4$ additional rounds, where $n$ is the size of
the network.  The algorithm is written in the locally shared memory
model with composite atomicity. It is proven assuming a distributed unfair
 daemon and takes $O(n\cdot\Delta^3)$ moves to stabilize,
where $\Delta$ is the degree of the network.

There are several other self-stabilizing solutions for particular
instances of $(f,g)$-alliances proposed in the locally
shared memory model with composite atomicity, {\em
  e.g.},~\cite{DingWS14,KakugawaM06,SrimaniX07,Turau07,PathanTXW12,YahiaouiBHK13}.

Algorithms given in ~\cite{SrimaniX07,PathanTXW12} work in anonymous
networks, however, they both
assume a central daemon. More precisely, Srimani and
Xu~\cite{SrimaniX07} give several algorithms which compute minimal
global offensive and $1$-minimal defensive
alliances in $O(n^3)$ moves.  Wang {\em et al.}~\cite{PathanTXW12} give
a self-stabilizing algorithm to compute a minimal $k$-dominating set
in $O(n^2)$ moves.

All other solutions~\cite{DingWS14,KakugawaM06,Turau07,YahiaouiBHK13}
consider arbitrary identified networks.  Turau~\cite{Turau07} gives a
self-stabilizing algorithm to compute a minimal dominating set in $9n$
moves, assuming a distributed unfair daemon.  Yahiaoui {\em et
  al.}~\cite{YahiaouiBHK13} give self-stabilizing algorithms to compute
a minimal global powerful alliance.  Their solution assumes a distributed unfair
 daemon and stabilizes in $O(n\cdot m)$ moves, where $m$ is
the number of edges in the network.

A safely converging self-stabilizing algorithm is given
in~\cite{KakugawaM06} for computing a minimal dominating set. The
algorithm first computes a (not necessarily minimal) dominating set in
$O(1)$ rounds and then safely stabilizes to a {\em minimal\/}
dominating set in $O(D)$ rounds, where $D$ is the diameter of
the network. However, a synchronous daemon is required.
A safely converging self-stabilizing algorithm
for computing minimal global offensive alliances
is given in~\cite{DingWS14}. This algorithm also assumes a synchronous daemon.
It first
computes a (not necessarily minimal) global offensive alliance within two
rounds, and then safely stabilizes to a {\em minimal\/} global
offensive alliance within $O(n)$ additional rounds.

To the best of our knowledge, until now there was no self-stabilizing
algorithm solving the 1-minimal $(f,g)$-alliance without any restriction
on $f$ and $g$.

\subsection{Algorithm \DA}

\paragraph{Overview.}

\begin{algorithm}
\small  
$~$ \\[0.3cm]
\textbf{Inputs:} \\[0.1cm]
\begin{tabular}{llll}
$\bullet$ & $\st_u \in \{ C, RB, RF \}$ & : & variable of \A\\  
$\bullet$ & $\PClean(u)$ & : & predicate of \A\\
$\bullet$ & $id_u$ & : & identifier of $u$, constant from the system
\end{tabular}
\\[0.3cm]
\textbf{Variables:} \\[0.1cm]
\begin{tabular}{llll}
$\bullet$ & $\VColor_u$ & : & Boolean\\
$\bullet$ & $\VSatisfaction_u \in \{-1,0,1\}$ & : & The score of $u$\\
$\bullet$ & $\VCanQuit_u$  & : & Boolean\\
$\bullet$ & $\VPointer_u \in \N[u] \cup \{ \bot \}$  & : & Closed Neighborhood Pointer
\end{tabular}
\\[0.3cm]
\textbf{Predicates:} \\[0.1cm]
\begin{tabular}{llll}
$\bullet$ & $\PICorrect(u)$ & $\equiv$ & $\satisfaction(u) \geq 0$ $\And$\\
          &                 &          & $[(\VSatisfaction_u = \satisfaction(u) = 1) \vee \VPointer_u=\bot \vee$\\
          &                 &          & $(\VPointer_u \neq \bot \wedge \VSatisfaction_u=1 \wedge \neg \VColor_{\VPointer_u})]$\\

$\bullet$ & $\PReset(u)$ &  $\equiv$ & $\VColor_u \And \VPointer_u = \bot \And \VCanQuit_u \And \VSatisfaction_u =1$\\

$\bullet$ & $\PQuitAlliance(u)$ &  $\equiv$ & $\VColor_u \wedge \numberYes(u) \geq f(u) \wedge (\forall v \in \N(u), \VSatisfaction_v =1)$\\
 $\bullet$ & $\PBecomeNormal(u)$ &  $\equiv$ & $\PQuitAlliance(u) \And (\forall v \in \N[u], \VPointer_v =u)$ \\

$\bullet$ & 
$\PRulePointer(u)$ &  $\equiv$ &  $\neg\PBecomeNormal(u) \And \VPointer_u \neq \BestPointer(u)$ 

\end{tabular}  
\\[0.3cm]
\textbf{Macros:}\\[0.1cm]
\begin{tabular}{llll}
$\bullet$ & $\reset(u)$ & : & $\VColor_u:= true$; $\VPointer_u := \bot$; $\VCanQuit_u := true$;
 $\VSatisfaction_u := 1$;\\
$\bullet$ & $\numberYes(u)$ & : & $|\{ w \in \N(u) ~|~ \VColor_w\}|$ 
\end{tabular} \\
$\bullet$ 
$
\left \{
   \begin{array}{l c l c r}
\satisfaction(u)=-1 & if  & \numberYes(u) < f(u) &  \wedge & \neg \VColor_u\\
\satisfaction(u)=-1 & if  & \numberYes(u) < g(u) & \wedge & \VColor_u\\
\satisfaction(u)=0 & if  & \numberYes(u) = f(u) & \wedge & \neg \VColor_u\\
\satisfaction(u)=0 & if  & \numberYes(u) = g(u) & \wedge & \VColor_u\\
\satisfaction(u)=1 & if & \numberYes(u) > f(u) & \wedge & \neg \VColor_u\\
\satisfaction(u)=1 & if & \numberYes(u) > g(u) & \wedge & \VColor_u\\    
   \end{array}
   \right .
$\\
\begin{tabular}{llll}
$\bullet$ & $\ComputeVariable(u)$ & : &  $\VSatisfaction_u := \satisfaction(u)$; $\VCanQuit_u := \PQuitAlliance(u)$;\\
$\bullet$ & $\BestPointer(u)$ & : & if $(\VSatisfaction_u \leq 0)$ return $\bot$;\\
          &                   &   & if $(\forall v \in \N[u], \neg \VCanQuit_u)$ then  return $\bot$;\\
          &                   &   & else  return $\argmin_{(v \in \N[u] | \VCanQuit_u)}(id_u)$;\\
$\bullet$ & $\update(u)$ & :  & $\ComputeVariable(u)$; $\VPointer_u := \BestPointer(u)$ ;
\end{tabular}
\\[0.3cm]
{\textbf{Rules:}} \\[0.1cm]
\begin{tabular}{lllll}
$\ruleColor(u)$ & : &  $\PClean(u) \And \PICorrect(u) \And$  & $\to$ & $\VColor_u := false$;\\
                &   &  $\PBecomeNormal(u)$                                                       &       & $\update(u)$;\\
\\
$\rulePointerA(u)$ & : & $\PClean(u) \And  \PICorrect(u) \And$ & $\to$ & $\VPointer_u := \bot$; \\
                   &  & $\PRulePointer(u) \And \VPointer_u \neq \bot$ &  &  $\ComputeVariable(u)$;\\
\\
$\rulePointerB(u)$ & : & $\PClean(u) \And  \PICorrect(u) \And$ & $\to$ &   $\update(u)$;\\
                   &    & $\PRulePointer(u) \And \VPointer_u = \bot$ &  & \\
\\
$\ruleCanQuitB(u)$ & : & $\PClean(u) \And \PICorrect(u) \And$& $\to$ &  $\ComputeVariable(u)$;\\
                   &    & $\neg\PBecomeNormal(u) \And \neg\PRulePointer(u) \And$ & & if $\satisfaction(u) \leq 0$ then\\
                   &    & $(\VSatisfaction_u \neq \satisfaction(u) \Or \VCanQuit_u \neq \PQuitAlliance(u))$ &  &  \qquad  $\VPointer_u = \bot$;
 \end{tabular}  
\\[0.3cm]
\caption{Algorithm \DA, code for every process $u$}
\label{alg:DAvar}
\end{algorithm}

Recall that we consider any network where $\delta_u \geq max(f(u),
g(u))$, for every process $u$.  Moreover, we assume that the
network is identified, meaning that each process $u$ can be
distinguished using a unique constant identifier, here noted $id_u$. 
The formal code of \DA is given in Algorithm \ref{alg:DAvar}. 
Informally, each process $u$ maintains the following four variables. 
\begin{description}
\item[$\VColor_u$:] a Boolean variable, the output of \DA. Process $u$  belongs to the $(f,g)$-alliance if and only if $\VColor_u$. 
\item[$\VSatisfaction_u$:] a variable, whose domain is $\{-1,0,1\}$. $\VSatisfaction_u \leq 0$
  if and only if no $u$'s neighbor can quit the alliance.
\item[$\VCanQuit_u$:] a Boolean variable. $\neg \VCanQuit_u$ if $u$
  cannot quit the alliance (in particular, if $u$ is out of
  the alliance).
\item[$\VPointer_u$:] a pointer variable, whose domain is $\N[u] \cup
  \{ \bot \}$. Either $\VPointer_u = \bot$ or $\VPointer_u$ designates
  the member of its closed neighborhood of smallest identifier such
  that $\VCanQuit_u$.
\end{description}

In the following, we assume that the system is initially in the configuration $\gamma_{init}$ where every process $u$ has the following local state:
\begin{center}
$\VColor_u = true$,  $\VSatisfaction_u =1$, $\VCanQuit_u = true$, 
$\VPointer_u = \bot$, $\st_u$ = $C$.
\end{center}
In particular, this means that all processes are initially in the
alliance. Then, the idea of the algorithm is reduced the alliance
until obtaining a 1-minimal $(f,g)$-alliance.  A process $u$ leaves
the alliance by executing $\ruleColor(u)$. To leave the alliance, $u$
should have enough neighbors in the alliance ($\numberYes(u) \geq
f(u)$), {\em approve} itself, and have
a {\em full} approval from all neighbors. Process $v$ approves $u$ if
$\VPointer_v = u$. Moreover, the approval of $v$ is {\em full} if
$\VSatisfaction_v = 1$.  Notice that, the $\VPointer$ pointers ensure
that removals from the alliance are locally central: in the closed
neighborhood of any process, at most one process leaves the alliance
at each step.

To ensure the liveness of the algorithm, a process $u$ gives its
approval (by executing $\rulePointerB(u)$, maybe preceded by $\rulePointerA(u)$)
to the member of its closed neighborhood 
having the smallest identifier among the ones requiring an approval
({\em i.e.}, the processes satisfying $\VCanQuit$).

To ensure that $\satisfaction(v) \geq 0$ is a closed predicate, 
a
process $v$ gives its approval to another process $u$ only if
$\satisfaction(v) = 1$ and none of its neighbor can leave the alliance
({\em i.e.}, $\VPointer_u \notin \N(v)$). This latter condition ensures
that no neighbor of $v$ leaves the alliance simultaneously to a new
approval of $v$. It is mandatory since otherwise the cause for which
$v$ gives its approval may be immediately outdated.  Hence, any
approval switching is done either in one step when the process leaves
the alliance, or in two atomic steps where $\VPointer_v$ first takes the value $\bot$ (rule $\rulePointerA(u)$) and then points to the suitable process, 
(rule $\rulePointerB(u)$).

Finally, the rule $\ruleCanQuitB(u)$ refreshes the values of
$\VSatisfaction_u$, $\VPointer_u$, and $\VCanQuit_u$ after a neighbor
left the alliance or updated its $\VSatisfaction$ variable.

\paragraph{Properties of \DA.}

Below, we show some properties of Algorithm \DA that will be
used for showing both its correctness and the self-stabilization of
its composition with Algorithm \A.

First, by checking the rules of \DA, we can remark that each
time a process $u$ sets $\VSatisfaction_u$ to a value other than 1, it
also sets $\VPointer_u$ to $\bot$, in the same step. Hence, by
construction we have the following lemma.

\begin{lemma}
\label{lem:VSatisfaction}
$\VSatisfaction_u=1$ $\Or$ $\VPointer_u=\bot$ is closed by \DA,
for every process $u$.
\end{lemma}

Since Algorithm \DA does not modify any variable from Algorithm \A, we have

\begin{remark}\label{lem:PClean:clos}
$\PClean(u)$ is closed by \DA, for every process $u$.
\end{remark}

\begin{lemma}
\label{lem:satisfaction}
Let $u$ be any process. Let $\gamma$ be a configuration where
$\PClean(u) \wedge \PICorrect(u)$ holds.  Let $\gamma'$ be any configuration such that $\gamma \mapsto \gamma'$. 
In $\gamma'$, $\satisfaction(u) \geq 0$.
\end{lemma}
\begin{proof}
  By definition of $\PBecomeNormal$, at most one process of $\N[u]$
  executes $\ruleColor$ in $\gamma \mapsto \gamma'$. If no process of
  $\N[u]$ executes $\ruleColor$, we are done. If 
  $\ruleColor(u)$ is executed, then $\numberYes(u) \geq f(u)$ in $\gamma$ (see
  $\PQuitAlliance(u)$) and so $\numberYes(u) \geq f(u)$ in $\gamma'$
  too and thus $\satisfaction(u) \geq 0$ in $\gamma'$.  Otherwise, let $v \in \N(u)$ such that $\ruleColor(v)$ is executed  in $\gamma \mapsto \gamma'$. In $\gamma$, $\VColor_v = true$ and
  $\PBecomeNormal(v)$ holds with, in particular, $\VPointer_u = v \neq
  \bot$, {\em i.e.}, $\VColor_{\VPointer_u}$ holds. Hence,
  $\PICorrect(u)$, $\VPointer_u \neq \bot$, and
  $\VColor_{\VPointer_u}$ imply $\satisfaction(u) = 1$ in $\gamma$,
  and so $\satisfaction(u) \geq 0$ in $\gamma'$.
\end{proof}

By definition of \DA, we have 
\begin{remark} Let $u$ be any process. Let $\gamma$ be a configuration where
$\VPointer_u=v$ with $v \neq u$.  Let $\gamma'$ be any configuration such that $\gamma \mapsto \gamma'$. 
In $\gamma'$, we have $\VPointer_u \in \{ v,  \bot \}$.
\end{remark}

\begin{lemma}\label{lem:PICorrect:clos:1}
Let $u$ be any process. Let $\gamma$ be a configuration where $\PClean(u) \wedge \PICorrect(u)$ holds. Let $\gamma'$ be any configuration such that $\gamma \mapsto \gamma'$ is a step  
where $v \in \N(u)$ executes $\ruleColor(v)$.
$\PICorrect(u)$ holds in $\gamma'$
\end{lemma}
\begin{proof}
  Since $\ruleColor(v)$ is enabled in $\gamma$, we have $\VPointer_u =
  v$ and $\VSatisfaction_u = 1$. So, if $u$ does not move in $\gamma
  \mapsto \gamma'$, we have $\VPointer_u \neq \bot$, $\VSatisfaction_u
  = 1$, and $\neg \VColor_{\VPointer_u}$. Hence, $\PICorrect(u)$ holds
  in $\gamma'$ by Lemma \ref{lem:satisfaction}.

Otherwise, either $\rulePointerA(u)$ or $\ruleCanQuitB(u)$ is executed
in $\gamma \mapsto \gamma'$. In the former case, $\VPointer_u = \bot$
in $\gamma'$ and by Lemma \ref{lem:satisfaction}, $\PICorrect(u)$
holds in $\gamma'$. In the latter case, in $\gamma'$ either
$\VPointer_u = v \neq \bot$, $\VSatisfaction_u = 1$ (by Lemma
\ref{lem:VSatisfaction}), and $\neg \VColor_{\VPointer_u}$, or
$\VPointer_u = \bot$. In either case, $\PICorrect(u)$ holds in
$\gamma'$ by Lemma \ref{lem:satisfaction}, and we are done.
\end{proof}

\begin{lemma}\label{lem:PICorrect:clos:2}
Let $u$ be any process. Let $\gamma$ be a configuration where $\PClean(u) \wedge \PICorrect(u)$ holds. Let $\gamma'$ be any configuration such that $\gamma \mapsto \gamma'$ is a step  
where $u$ executes an action and none of its neighbor executes $\ruleColor$.
$\PICorrect(u)$ holds in $\gamma'$.
\end{lemma}
\begin{proof}
  Since no neighbor of $u$ executes $\ruleColor$ in $\gamma \mapsto
  \gamma'$, we have $\VSatisfaction_u = \satisfaction(u)$ in $\gamma'$,
  and by Lemmas \ref{lem:VSatisfaction} and \ref{lem:satisfaction}, we
  are done.
\end{proof}

By Remark~\ref{lem:PClean:clos} and Lemmas \ref{lem:PICorrect:clos:1}-\ref{lem:PICorrect:clos:2}, follows.

\begin{corollary}\label{coro:normal}
  $\PClean(u) \wedge \PICorrect(u)$ is closed by \DA, for
  every process $u$.
\end{corollary}

Since $u$ is disabled in \DA if $\neg \PClean(u)$ holds.  $\neg \PClean(u)
\wedge \PICorrect(u)$ is also closed by \DA, for every process
$u$. Hence, follows.

\begin{corollary}\label{coro:normal2}
  $\PICorrect(u)$ is closed by \DA, for
  every process $u$.
\end{corollary}

\paragraph{Partial Correctness.}

As for Algorithm \DU, we focus on configurations of \DA satisfying
$\PICorrect(u) \wedge \PClean(u)$, every process $u$. Indeed, again,
$\gamma_{init}$ belongs to this class of configurations, and a
configuration of \DA $\circ$ \A is normal if and only if $\PClean(u)
\And \PICorrect(u)$ holds for every process $u$. So, the properties we
exhibit now will be, in particular, satisfied at the completion of \A.

Consider any terminal configuration of \DA where $\PClean(u)
\wedge \PICorrect(u)$ holds for every process $u$.  By checking the
code of Algorithm \DA, we can remark that every process $u$
satisfies \begin{center} $\VSatisfaction_u = \satisfaction(u)$ $\And$
  $\VCanQuit_u = \PQuitAlliance(u)$ $\And$ \\
  $\VPointer_u = \BestPointer(u)$ $\And$ $\neg\PBecomeNormal(u)$.
\end{center} Based on this, one can easily establish that in such a
terminal configuration, the set $A = \{u \in V\ |\ \VColor_u\}$ is a
1-minimal $(f,g)$ alliance of the network.
Indeed, for every process $u$, since $\PICorrect(u)$ holds, $\satisfaction(u) \geq 0$, which in turn implies that $A$ is an $(f,g)$ alliance.
Assume then, by the contradiction, that $A$ is not 1-minimal in some
terminal configuration of \DA where $\PClean(u) \wedge
\PICorrect(u)$ holds for every process $u$. Let $m$ be the process of
minimum identifier such that $A - \{m\}$ is an $(f,g)$
alliance. First, by definition, $\VColor_m \wedge \numberYes(m) \geq
f(m)$ holds. Then, $\forall u \in \N[m]$, $\VSatisfaction_u =
\satisfaction(u) = 1$ since $A - \{m\}$ is an $(f,g)$ alliance. So,
$\PQuitAlliance(m)$ holds, which implies that $\VCanQuit_m =
true$. Finally, by minimality of the $m$'s identifier, $\forall u \in
\N[m]$, $\VPointer_u = \BestPointer(u) = m$. Hence, $\PBecomeNormal(m)$
holds, which in turn implies that $\ruleColor(m)$ is enabled, a
contradiction. Hence, follows.

\begin{theorem}\label{theo:term:alliance}
  In any terminal configuration of \DA where
  $\PClean(u) \wedge \PICorrect(u)$ holds for every process $u$, the set
  $A = \{u \in V\ |\ \VColor_u\}$ is a 1-minimal $(f,g)$ alliance of the
  network.
\end{theorem}

Consider now any execution $e$ of \DA starting from
$\gamma_{init}$. In $\gamma_{init}$, we have $\PClean(u) \wedge
\PICorrect(u)$ for every process $u$. Then, $\PClean(u) \wedge
\PICorrect(u)$, for every process $u$, is invariant in $e$, by
Corollary~\ref{coro:normal}. Hence, we have the following corollary.

\begin{corollary}\label{coro:term:alliance}
  Let $e$ be any execution  of \DA starting from $\gamma_{init}$. If $e$
  terminates, then the set $A = \{u \in V\ |\ \VColor_u\}$ is a 1-minimal
  $(f,g)$ alliance of the network in the terminal configuration of $e$.
\end{corollary}

\paragraph{Termination.}

We now show that any execution of \DA (starting from any
arbitrary configuration) eventually terminates.
Let $v$ be any process. Let $e$ be any execution of \DA.  First, $\ruleColor(v)$ switches $\VColor_v$
from true to false and no rule of \DA sets $\VColor_v$ from
false to true. So,
\begin{itemize}
\item[(1)] $\ruleColor(v)$ is executed at most once in $e$.
\end{itemize}
Moreover, this implies that
\begin{itemize}
\item[(2)]
 the value of the macro 
$\numberYes(v)$ is monotonically
non-increasing in $e$.
\end{itemize}
If  $\satisfaction(u) < 0$ holds for some process $u$ in some configuration of $e$, then $\PICorrect(u)$ does not hold and so $u$ is disabled. Moreover, by (2), $u$ is disabled forever in $e$. Assume, otherwise, that  $\satisfaction(u) \geq
0$. Then, $\satisfaction(u)$ may increase at most once in $e$:
when $\ruleColor(u)$ is executed while $\numberYes(u) > f(u)$.
  So,
\begin{itemize}
\item[(3)] Every process $u$ updates the value of $\VSatisfaction_u$
  at most $4$ times in $e$.
\end{itemize}
Hence, overall (by (1)-(3)), 
the value of $\PQuitAlliance(v)$ changes
at most $4\delta_v+2$ in $e$, and thus
\begin{itemize}
\item[(4)]
$v$ updates the value of $\VCanQuit_v$ at most $4\delta_v+3$
in $e$.
\end{itemize}
By (3) and (4) 
\begin{itemize}
\item[(5)]
$\ruleCanQuitB(v)$ is executed at most
most $4\delta_v+7$ times in $e$.
\end{itemize}
The value of $\BestPointer(v)$ may change only when a process $u$ in the
closed neighborhood of $v$ changes the value of its variable
$\VCanQuit_u$ or when $v$ updates $\VSatisfaction_v$. 
So, by (3) and (4),
 the value of
$\BestPointer(v)$ changes at most $\delta_v.(4\Delta+3)+4\delta_v+7$ times.
So, 
\begin{itemize}
\item[(6)]
$v$ executes $\rulePointerA$ and $\rulePointerB$ at most
$\delta_v.(4\Delta+3)+4\delta_v+8$ times each in $e$.
\end{itemize}

Overall (by (1), (5), and (6)), follows. 
\begin{lemma}\label{lem:move:alliance}
A process $v$ executes at most $8\delta_v\Delta+18\delta_v+24$ moves in an execution of \DA.
\end{lemma}

\begin{corollary}\label{coro:move:alliance} Any execution of \DA contains at most $16.\Delta.m+36.m+24.n$ moves,
  {\em i.e.}, $O(\Delta.m)$ moves.
\end{corollary}

By Corollaries~\ref{coro:term:alliance} and~\ref{coro:move:alliance} , we can conclude with the following
theorem.

\begin{theorem} \DA is distributed (non self-stabilizing)
  1-minimal $(f,g)$-alliance algorithm which terminates in at most
  $O(\Delta.m)$ moves.
\end{theorem}

\paragraph{Round Complexity.}

We already know that in any execution of \DA, each process
executes $\ruleColor$ at most once. So, along any execution there are
at most $n$ steps containing the execution of some 
$\ruleColor$. We now say that a step is {\em $Color$-restricted}, if
no rule $\ruleColor$ is executed during that step. Similarly, we say
that a round is {\em $Color$-restricted} if it only consists of
$Color$-restricted steps.  In the sequel, we show that any execution that 
starts from a configuration where $\PClean(u) \wedge \PICorrect(u)$ holds for every process $u$
contains at most $4$ consecutive $Color$-restricted rounds. Hence, the
number of rounds in any execution starting from $\gamma_{init}$ is bounded by $5n+4$.  To that goal,
we first specialize the notion of closure.  A predicate $P$ over
configurations of \DA is {\em $Color$-restricted closed} if for
every $Color$-restricted step $\gamma \mapsto \gamma'$, $P(\gamma)
\Rightarrow P(\gamma')$.

We then consider the following predicates over configurations of \DA.
\begin{itemize}
  \item $\mathcal{P}_5$ is true if and only if every process $u$
    satisfies $\PClean(u) \wedge \PICorrect(u)$.
  \item $\mathcal{P}_6$ is true if and only if  $\mathcal{P}_5$ holds and every process $u$
    satisfies $\VSatisfaction_u = \satisfaction(u)$.
  \item
    $\mathcal{P}_7$ is true if and only if $\mathcal{P}_6$ holds and
    every process $u$ satisfies $\VCanQuit_u = \PQuitAlliance(u)$.
  \item $\mathcal{P}_8$ is true if and only if $\mathcal{P}_7$ holds
    and every process $u$ satisfies $\VPointer_u \in
    \{\BestPointer(u), \bot\}$.
   \item $\mathcal{P}_9$ is true if and only if $\mathcal{P}_8$ holds
     and every process $u$ satisfies $\VPointer_u = \BestPointer(u)$.
\end{itemize}

\begin{lemma}
  \label{lem:alliance-round1}
  $\mathcal{P}_6$ is $Color$-restricted closed. Moreover, after one
  $Color$-restricted round from any configuration satisfying  $\mathcal{P}_5$, a configuration
  satisfying $\mathcal{P}_6$ is reached.
\end{lemma}
\begin{proof}
  Let $u$ be a process.  The value of $\satisfaction(u)$ stays
  unchanged during a $Color$-restricted step. So, the predicate
  $\VSatisfaction_u = \satisfaction(u)$ is $Color$-restricted closed,
  and so $\mathcal{P}_6$ is.  Let $\gamma$ be a configuration
  satisfying $\mathcal{P}_5$. Recall that $\mathcal{P}_5$ is closed,
  by Corollary~\ref{coro:normal}. So, $\PClean(u) \wedge
  \PICorrect(u)$ holds forever from $\gamma$.  If
  $\VSatisfaction_u \neq \satisfaction(u)$ in $\gamma$, then $u$ is
  enabled in \DA. Now, if $u$ moves in a $Color$-restricted step, then
  $\VSatisfaction_u = \satisfaction(u)$ in the reached
  configuration. Hence, $\VSatisfaction_u = \satisfaction(u)$ holds,
  for every process $u$, within at most one round from $\gamma$, and
  we are done.
\end{proof}

\begin{lemma}
  \label{lem:alliance-round2}
  $\mathcal{P}_7$ is $Color$-restricted closed. Moreover, after one
  $Color$-restricted round from any configuration satisfying
  $\mathcal{P}_6$, a configuration satisfying $\mathcal{P}_7$ is
  reached.
\end{lemma}
\begin{proof}
  Let $u$ be a process.  Let $\gamma_i \mapsto \gamma_{i+1}$ be a
  $Color$-restricted step such that $\mathcal{P}_6(\gamma_i)$ holds.
  For every process $v$, the value of $\VColor_v$, $\VSatisfaction_v$,
  and $\numberYes(v)$ stay unchanged during $\gamma_i \mapsto
  \gamma_{i+1}$. Therefore, the value of $\PQuitAlliance(u)$ stays
  unchanged during $\gamma_i \mapsto \gamma_{i+1}$ for every process
  $u$.  So, the predicate $\mathcal{P}_6 \wedge \VCanQuit_u =
  \PQuitAlliance(u)$ is $Color$-restricted closed, and so
  $\mathcal{P}_7$ is.
  Assume now that, $\VCanQuit_u \neq \PQuitAlliance(u)$ holds in
  $\gamma_i$.  Then, $u$ is enabled in $\gamma_i$ and if $u$ moves in
  $\gamma_i \mapsto \gamma_{i+1}$, $\VCanQuit_u = \PQuitAlliance(u)$
  holds in  in $\gamma_{i+1}$.  Indeed, the value of
  $\PQuitAlliance(u)$ stays unchanged during the step. Therefore,
  after at most one $Color$-restricted round from $\gamma_i$, we have
  $\VCanQuit_u = \PQuitAlliance(u)$ for every process $u$, and so
  $\mathcal{P}_7$ holds.
\end{proof} 

\begin{lemma}
  \label{lem:alliance-round3}
  $\mathcal{P}_8$ is $Color$-restricted closed. Moreover, after one
  $Color$-restricted round from any configuration satisfying
  $\mathcal{P}_7$, a configuration satisfying $\mathcal{P}_8$ is
  reached.
\end{lemma}
\begin{proof}
Let $u$ be a process.  Let $\gamma_i \mapsto \gamma_{i+1}$ be a
$Color$-restricted step such that $\mathcal{P}_7(\gamma_i)$ holds.
During that step, only rules $\rulePointerA$ or $\rulePointerB$ are
executed. Moreover, for every process $v$, the value of $\VColor_v$
and $\VCanQuit_v$ stay unchanged during $\gamma_i \mapsto
\gamma_{i+1}$.  So, the value of $\BestPointer(u)$ stays unchanged as
well.  So, the predicate $\mathcal{P}_7 \wedge \VPointer_u \in
\{\BestPointer(u), \bot\}$ is $Color$-restricted closed, and so
$\mathcal{P}_8$ is.
 Assume now that, $\VPointer_u \notin \{\BestPointer(u), \bot\}$ holds
 in $\gamma_i$.  Then, $u$ is enabled in $\gamma_i$ and if $u$ moves
 in $\gamma_i \mapsto \gamma_{i+1}$, $\VPointer_u \in
 \{\BestPointer(u), \bot\}$ in $\gamma_{i+1}$.  Indeed, the value of
 $\BestPointer(u)$ stays unchanged during the step. Therefore, after
 at most one $Color$-restricted round from $\gamma_i$, we have
 $\VPointer_u \in \{\BestPointer(u), \bot\}$ for every process $u$,
 and so $\mathcal{P}_8$ holds.
\end{proof}

\begin{lemma}
  \label{lem:alliance-round4}
  $\mathcal{P}_9$ is $Color$-restricted closed. Moreover, after one
  $Color$-restricted round from any configuration satisfying
  $\mathcal{P}_8$, a configuration satisfying $\mathcal{P}_9$ is
  reached.
\end{lemma}
\begin{proof}
Let $u$ be a process.  Let $\gamma_i \mapsto \gamma_{i+1}$ be a
$Color$-restricted step such that $\mathcal{P}_8(\gamma_i)$ holds.
During that step, only rules $\rulePointerB$ are
executed. Moreover, for every process $v$, the value of $\VColor_v$
and $\VCanQuit_v$ stay unchanged during $\gamma_i \mapsto
\gamma_{i+1}$.  So, the value of $\BestPointer(u)$ stays unchanged as
well.  So, the predicate $\mathcal{P}_8 \wedge \VPointer_u =
\BestPointer(u)$ is $Color$-restricted closed, and so
$\mathcal{P}_9$ is.
 Assume now that, $\VPointer_u \neq \BestPointer(u)$ holds
 in $\gamma_i$.  Then, $u$ is enabled in $\gamma_i$ and if $u$ moves
 in $\gamma_i \mapsto \gamma_{i+1}$, $\VPointer_u =
 \BestPointer(u)$ in $\gamma_{i+1}$.  Indeed, the value of $\BestPointer(u)$
 stays unchanged during the step. Therefore, after at most one
 $Color$-restricted round from $\gamma_i$, we have $\VPointer_u =
 \BestPointer(u)$ for every process $u$, and so
 $\mathcal{P}_8$ holds.
\end{proof}

\begin{theorem}
\label{theo:alliance-round}
Starting from any configuration satisfying  $\mathcal{P}_5$, Algorithm \DA terminates in at most $5n+4$ rounds.
\end{theorem}
\begin{proof}
According to Lemmas \ref{lem:alliance-round1}-\ref{lem:alliance-round4}, after any $4$ consecutive $Color$-restricted rounds, every process $u$ satisfies:
\begin{multline*}
\VSatisfaction_u = \satisfaction(u) ~\And~
\VCanQuit_u = \PQuitAlliance(u)~\And 
\neg\PRulePointer(u)
\end{multline*}
So only  $\ruleColor(u)$ may be enabled at $u$. We conclude
that an execution of $\ruleColor$ occurs at least every $5$ rounds,
unless the system reaches a terminal configuration.  Since, along any
execution, there are at most $n$ steps containing the execution of
some $\ruleColor$, the theorem follows.
\end{proof}

Since $\gamma_{init}$ satisfies $\mathcal{P}_5$, we have the following
corollary.

\begin{corollary} Starting from $\gamma_{init}$, Algorithm \DA terminates in at most $5n+4$ rounds.
  \end{corollary}

\subsection{Algorithm \DA $\circ$ \A}

\paragraph{Requirements.} To show the self-stabilization of \DA
$\circ$ \A, we should first establish that \DA meets the requirements
~\ref{RQ1} to~\ref{RQ6}, given in Subsection~\ref{sect:require}.

\begin{enumerate}
\item From the code of \DA, we can deduce that
  Requirements~\ref{RQ1},~\ref{RQ4},~\ref{RQ3}, and ~\ref{RQ5} are  satisfied.
\item Requirement~\ref{RQ2} is satisfied since $\PICorrect(u)$ does
  not involve any variable of \A and is closed by \DA (Corollary~\ref{coro:normal2}).
\item Finally, recall that $\delta_u \geq max(f(u), g(u))$, for every
  process $u$. So, if $\PReset(v)$ holds, for every $v \in \N[u]$, then $\satisfaction(u) = 1$ and so
  $\PICorrect(u)$ holds, by definition. Hence, Requirement~\ref{RQ6}
  holds.
\end{enumerate}

\paragraph{Partial Correctness.}

Let $\gamma$ be any terminal configuration of \DA $\circ$
\A. Then, $\gamma_{|\A}$ is a terminal configuration of \A and, by
Theorem~\ref{theo:termA}, $\PClean(u)
\And \PICorrect(u)$ holds in $\gamma$, and so $\gamma_{|\A}$, for
every process $u$. Moreover, $\gamma_{|\DA}$ is a terminal
configuration of \DA. Hence, by
Theorem~\ref{theo:term:alliance}, follows.

\begin{theorem}\label{theo:term:alliance:2}
  In any terminal configuration of \DA $\circ$ \A. the set $\{u
  \in V\ |\ \VColor_u\}$ is a 1-minimal $(f,g)$ alliance of the
  network.
\end{theorem}

\paragraph{Termination and Self-Stabilization.}

Let $u$ be a process. Let $e$ be an execution of \DA $\circ$ \A.
By Corollary \ref{cor:seg} (page \pageref{cor:seg}), the sequence of
rules executed by a process $u$ in a segment of $e$ belongs to the
following language:
$$(\ruleC + \varepsilon)\ \wordsI\ (\ruleRB + \ruleR + \varepsilon)\ (\ruleRF + \varepsilon)$$ where $\wordsI$ be any sequence of rules of \DA.
Moreover, by Lemma~\ref{lem:transfert} (page \pageref{lem:transfert}) and 
Lemma~\ref{lem:move:alliance}, $\wordsI$ is bounded by
$8\delta_u\Delta+18\delta_u+24$.
Thus, $u$ executes at most
$8\delta_u\Delta+18\delta_u+27$ moves in any segment of
$e$ and, overall each segment of $e$ contains at most
$16m\Delta+36m+27n$ moves. Since, $e$ contains at most $n+1$ segments
(Remark~\ref{rem:nbseg}, page \pageref{rem:nbseg}), $e$ contains at most
$(n+1).(16m\Delta+36m+27n)$ moves, and we have the following theorem.

\begin{theorem}\label{theo:move:ss}
 Any execution of \DA $\circ$ \A terminates in $O(\Delta.n.m)$ moves. 
\end{theorem}

By Theorems~\ref{theo:term:alliance:2} and~\ref{theo:move:ss}, we can conclude:

\begin{theorem} \DA $\circ$ \A is
  self-stabilizing for the 1-minimal $(f,g)$-alliance problem. Its
  stabilization time is in $O(\Delta.n.m)$ moves.
\end{theorem}

\paragraph{Round Complexity.}

Corollary \ref{cor:round} (page \pageref{cor:round}) establishes that
after at most $3n$ rounds a normal configuration of \DA $\circ$ \A is
reached. Then, since the set of normal configuration is closed
(still by Corollary \ref{cor:round}), all rules of \A algorithm are
disabled forever from such a configuration, by Lemma~\ref{lem:a4:term}
(page \pageref{lem:a4:term}).  Moreover, in a normal configuration,
$\PClean(u) \And \PICorrect(u)$ holds, for every process $u$ (still
by Corollary \ref{cor:round}).  Hence, after at most $5n+4$ additional
rounds, a terminal configuration of \DA $\circ$ \A is reached, by
Theorem \ref{theo:alliance-round}, and follows.

\begin{theorem}
The
stabilization time of \DA $\circ$ \A is  at most $8n+4$ rounds.
\end{theorem}

\bibliographystyle{plain}
\bibliography{biblio}

\begin{thebibliography}{10}

\bibitem{AB98j}
Yehuda Afek and Anat Bremler{-}Barr.
\newblock Self-stabilizing unidirectional network algorithms by power supply.
\newblock {\em Chicago J. Theor. Comput. Sci.}, 1998, 1998.

\bibitem{ACDDP14}
Karine Altisen, Alain Cournier, St{\'{e}}phane Devismes, Ana{\"{\i}}s Durand,
  and Franck Petit.
\newblock Self-stabilizing leader election in polynomial steps.
\newblock {\em Inf. Comput.}, 254:330--366, 2017.

\bibitem{AngluinAER07}
Dana Angluin, James Aspnes, David Eisenstat, and Eric Ruppert.
\newblock The computational power of population protocols.
\newblock {\em Distributed Computing}, 20(4):279--304, 2007.

\bibitem{AG94j}
Anish Arora and Mohamed~G. Gouda.
\newblock Distributed reset.
\newblock {\em {IEEE} Trans. Computers}, 43(9):1026--1038, 1994.

\bibitem{AO94c}
Baruch Awerbuch and Rafail Ostrovsky.
\newblock Memory-efficient and self-stabilizing network \{RESET\} (extended
  abstract).
\newblock In James~H. Anderson, David Peleg, and Elizabeth Borowsky, editors,
  {\em Proceedings of the Thirteenth Annual {ACM} Symposium on Principles of
  Distributed Computing, Los Angeles, California, USA, August 14-17, 1994},
  pages 254--263. {ACM}, 1994.

\bibitem{APV91c}
Baruch Awerbuch, Boaz Patt{-}Shamir, and George Varghese.
\newblock Self-stabilization by local checking and correction (extended
  abstract).
\newblock In {\em 32nd Annual Symposium on Foundations of Computer Science, San
  Juan, Puerto Rico, 1-4 October 1991}, pages 268--277. {IEEE} Computer
  Society, 1991.

\bibitem{APVD94c}
Baruch Awerbuch, Boaz Patt{-}Shamir, George Varghese, and Shlomi Dolev.
\newblock Self-stabilization by local checking and global reset (extended
  abstract).
\newblock In {\em Distributed Algorithms, 8th International Workshop, {WDAG}
  '94, Terschelling, The Netherlands, September 29 - October 1, 1994,
  Proceedings}, volume 857 of {\em Lecture Notes in Computer Science}, pages
  326--339. Springer, 1994.

\bibitem{BenOthman2013199}
Jalel Ben-Othman, Karim Bessaoud, Alain Bui, and Laurence Pilard.
\newblock Self-stabilizing algorithm for efficient topology control in wireless
  sensor networks.
\newblock {\em Journal of Computational Science}, 4(4):199 -- 208, 2013.

\bibitem{berge2001theory}
C.~Berge.
\newblock {\em The Theory of Graphs}.
\newblock Dover books on mathematics. Dover, 2001.

\bibitem{B07t}
Christian Boulinier.
\newblock {\em L'Unisson}.
\newblock PhD thesis, Universit{\'{e}} de Picardie Jules Vernes, France, 2007.

\bibitem{BPV04c}
Christian Boulinier, Franck Petit, and Vincent Villain.
\newblock When graph theory helps self-stabilization.
\newblock In {\em Proceedings of the Twenty-Third Annual {ACM} Symposium on
  Principles of Distributed Computing ({PODC}), July 25-28, 2004}, pages
  150--159, 2004.

\bibitem{BuiDPV99}
Alain Bui, Ajoy~Kumar Datta, Franck Petit, and Vincent Villain.
\newblock Optimal {PIF} in tree networks.
\newblock In Yuri Breitbart, Sajal~K. Das, Nicola Santoro, and Peter Widmayer,
  editors, {\em Distributed Data {\&} Structures 2, Records of the 2nd
  International Meeting {(WDAS} 1999), Princeton, USA, May 10-11, 1999},
  volume~6 of {\em Proceedings in Informatics}, pages 1--16. Carleton
  Scientific, 1999.

\bibitem{Caron20131533}
Eddy Caron, Florent Chuffart, and Cédric Tedeschi.
\newblock When self-stabilization meets real platforms: An experimental study
  of a peer-to-peer service discovery system.
\newblock {\em Future Generation Computer Systems}, 29(6):1533 -- 1543, 2013.

\bibitem{CaronDDL10}
Eddy Caron, Ajoy~Kumar Datta, Benjamin Depardon, and Lawrence~L. Larmore.
\newblock A self-stabilizing k-clustering algorithm for weighted graphs.
\newblock {\em J. Parallel Distrib. Comput.}, 70(11):1159--1173, 2010.

\bibitem{DBLP:journals/ppl/CaronDPT10}
Eddy Caron, Fr{\'{e}}d{\'{e}}ric Desprez, Franck Petit, and C{\'{e}}dric
  Tedeschi.
\newblock Snap-stabilizing prefix tree for peer-to-peer systems.
\newblock {\em Parallel Processing Letters}, 20(1):15--30, 2010.

\bibitem{CDDLR15j}
Fabienne Carrier, Ajoy~Kumar Datta, St{\'{e}}phane Devismes, Lawrence~L.
  Larmore, and Yvan Rivierre.
\newblock Self-stabilizing (f, g)-alliances with safe convergence.
\newblock {\em J. Parallel Distrib. Comput.}, 81-82:11--23, 2015.

\bibitem{CYH91}
Nian{-}Shing Chen, Hwey{-}Pyng Yu, and Shing{-}Tsaan Huang.
\newblock A self-stabilizing algorithm for constructing spanning trees.
\newblock {\em Inf. Process. Lett.}, 39(3):147--151, 1991.

\bibitem{CDPV02c}
Alain Cournier, Ajoy~K. Datta, Franck Petit, and Vincent Villain.
\newblock Snap-stabilizing {PIF} algorithm in arbitrary networks.
\newblock In {\em Proceedings of the 22nd International Conference on
  Distributed Computing Systems (ICDCS'02)}, pages 199--206, 2002.

\bibitem{CouvreurFG92}
Jean{-}Michel Couvreur, Nissim Francez, and Mohamed~G. Gouda.
\newblock Asynchronous unison (extended abstract).
\newblock In {\em Proceedings of the 12th International Conference on
  Distributed Computing Systems, Yokohama, Japan, June 9-12, 1992}, pages
  486--493. {IEEE} Computer Society, 1992.

\bibitem{CFG92c}
Jean{-}Michel Couvreur, Nissim Francez, and Mohamed~G. Gouda.
\newblock Asynchronous unison (extended abstract).
\newblock In {\em Proceedings of the 12nd International Conference on
  Distributed Computing Systems (ICDCS'92)}, pages 486--493, 1992.

\bibitem{DLP11}
Ajoy~K. Datta, Lawrence~L. Larmore, and Priyanka Vemula.
\newblock An ${O}({N})$-time self-stabilizing leader election algorithm.
\newblock {\em JPDC}, 71(11):1532--1544, 2011.

\bibitem{DJ16}
St\'ephane Devismes and Colette Johnen.
\newblock Silent self-stabilizing {BFS} tree algorithms revisited.
\newblock {\em JPDC}, 97:11 -- 23, 2016.

\bibitem{DP12}
St{\'{e}}phane Devismes and Franck Petit.
\newblock On efficiency of unison.
\newblock In {\em 4th Workshop on Theoretical Aspects of Dynamic Distributed
  Systems, ({TADDS} '12)}, pages 20--25, 2012.

\bibitem{D74j}
Edsger~W. Dijkstra.
\newblock {Self-stabilizing Systems in Spite of Distributed Control}.
\newblock {\em Commun. ACM}, 17(11):643--644, 1974.

\bibitem{DingWS14}
Yihua Ding, James~Z. Wang, and Pradip~K. Srimani.
\newblock Self-stabilizing minimal global offensive alliance algorithm with
  safe convergence in an arbitrary graph.
\newblock In T.~V. Gopal, Manindra Agrawal, Angsheng Li, and S.~Barry Cooper,
  editors, {\em Theory and Applications of Models of Computation - 11th Annual
  Conference (TAMC)}, volume 8402 of {\em Lecture Notes in Computer Science},
  pages 366--377, Chennai, India, April 11-13 2014. Springer.

\bibitem{DIM93}
S~Dolev, A~Israeli, and S~Moran.
\newblock Self-stabilization of dynamic systems assuming only {R}ead/{W}rite
  atomicity.
\newblock {\em {Distributed Computing}}, 7(1):3--16, 1993.

\bibitem{Dolev1997122}
Shlomi Dolev.
\newblock Self-stabilizing routing and related protocols.
\newblock {\em Journal of Parallel and Distributed Computing}, 42(2):122 --
  127, 1997.

\bibitem{DolevGS96}
Shlomi Dolev, Mohamed~G. Gouda, and Marco Schneider.
\newblock Memory requirements for silent stabilization.
\newblock {\em Acta Informatica}, 36(6):447--462, 1999.

\bibitem{DouradoPRS11}
Mitre~Costa Dourado, Lucia~Draque Penso, Dieter Rautenbach, and Jayme~Luiz
  Szwarcfiter.
\newblock The south zone: Distributed algorithms for alliances.
\newblock In Xavier D{\'e}fago, Franck Petit, and Vincent Villain, editors,
  {\em Stabilization, Safety, and Security of Distributed Systems - 13th
  International Symposium, (SSS)}, volume 6976 of {\em Lecture Notes in
  Computer Science}, pages 178--192, Grenoble, France, October 10-12 2011.
  Springer.

\bibitem{GareyJ79}
M.~R. Garey and David~S. Johnson.
\newblock {\em Computers and Intractability: A Guide to the Theory of
  NP-Completeness}.
\newblock W. H. Freeman, 1979.

\bibitem{GHIJ14}
Christian Glacet, Nicolas Hanusse, David Ilcinkas, and Colette Johnen.
\newblock Disconnected components detection and rooted shortest-path tree
  maintenance in networks.
\newblock In {\em SSS'14}, pages 120--134, 2014.

\bibitem{GuptaMOR05}
Anupam Gupta, Bruce~M. Maggs, Florian Oprea, and Michael~K. Reiter.
\newblock Quorum placement in networks to minimize access delays.
\newblock In Marcos~Kawazoe Aguilera and James Aspnes, editors, {\em
  Proceedings of the Twenty-Fourth Annual ACM Symposium on Principles of
  Distributed Computing (PODC)}, pages 87--96, Las Vegas, NV, USA, July 17-20
  2005. ACM.

\bibitem{HC92}
Shing-Tsaan Huang and Nian-Shing Chen.
\newblock A self-stabilizing algorithm for constructing breadth-first trees.
\newblock {\em IPL}, 41(2):109--117, 1992.

\bibitem{HuangC93}
Shing{-}Tsaan Huang and Nian{-}Shing Chen.
\newblock Self-stabilizing depth-first token circulation on networks.
\newblock {\em Distributed Computing}, 7(1):61--66, 1993.

\bibitem{KakugawaM06}
Hirotsugu Kakugawa and Toshimitsu Masuzawa.
\newblock A self-stabilizing minimal dominating set algorithm with safe
  convergence.
\newblock In {\em 20th International Parallel and Distributed Processing
  Symposium (IPDPS)}, pages 8.--, Rhodes Island, Greece, 25-29 April 2006.
  IEEE.

\bibitem{KP93j}
Shmuel Katz and Kenneth~J. Perry.
\newblock Self-stabilizing extensions for message-passing systems.
\newblock {\em Distributed Computing}, 7(1):17--26, 1993.

\bibitem{YNKM18c}
Yonghwan Kim, Junya Nakamura, Yoshiaki Katayama, and Toshimitsu Masuzawa.
\newblock A cooperative partial snapshot algorithm for checkpoint-rollback
  recovery of large-scale and dynamic distributed systems.
\newblock In {\em CANDAR'18}, pages 285--291, 11 2018.

\bibitem{LiaoC03}
Chung-Shou Liao and Gerard~J. Chang.
\newblock k-tuple domination in graphs.
\newblock {\em Inf. Process. Lett.}, 87(1):45--50, 2003.

\bibitem{Telematik_SSS_2013_Neighborhood}
Gerry Siegemund, Volker Turau, Christoph Weyer, Stefan Lobs, and Jörg Nolte.
\newblock Brief announcement: Agile and stable neighborhood protocol for wsns.
\newblock In {\em Proceedings of the 15th International Symposium on
  Stabilization, Safety, and Security of Distributed Systems (SSS'13)}, pages
  376--378, November 2013.

\bibitem{Sigarreta2009219}
J.M. Sigarreta and J.A. Rodr\'iguez.
\newblock On the global offensive alliance number of a graph.
\newblock {\em Discrete Applied Mathematics}, 157(2):219 -- 226, 2009.

\bibitem{SigarretaR06}
Jose~Maria Sigarreta and Juan~Alberto Rodr\'{\i}guez-Velazquez.
\newblock On defensive alliances and line graphs.
\newblock {\em Appl. Math. Lett.}, 19(12):1345--1350, 2006.

\bibitem{Sloman:1987}
M~Sloman and J~Kramer.
\newblock {\em Distributed systems and computer networks.}
\newblock Prentice Hall, 1987.

\bibitem{SrimaniX07}
Pradip~K. Srimani and Zhenyu Xu.
\newblock Distributed protocols for defensive and offensive alliances in
  network graphs using self-stabilization.
\newblock In {\em International Conference on Computing: Theory and
  Applications (ICCTA)}, pages 27--31, Kolkata, India, March 5-7 2007. IEEE
  Computer Society.

\bibitem{Turau07}
Volker Turau.
\newblock Linear self-stabilizing algorithms for the independent and dominating
  set problems using an unfair distributed scheduler.
\newblock {\em Inf. Process. Lett.}, 103(3):88--93, 2007.

\bibitem{PathanTXW12}
Guangyuan Wang, Hua Wang, Xiaohui Tao, and Ji~Zhang.
\newblock A self-stabilizing algorithm for finding a minimal k-dominating set
  in general networks.
\newblock In Yang Xiang, Mukaddim Pathan, Xiaohui Tao, and Hua Wang, editors,
  {\em Data and Knowledge Engineering}, Lecture Notes in Computer Science,
  pages 74--85. Springer Berlin Heidelberg, 2012.

\bibitem{YahiaouiBHK13}
Sa\"{\i}d Yahiaoui, Yacine Belhoul, Mohammed Haddad, and Hamamache Kheddouci.
\newblock Self-stabilizing algorithms for minimal global powerful alliance sets
  in graphs.
\newblock {\em Inf. Process. Lett.}, 113(10-11):365--370, 2013.

\end{thebibliography}

\end{document}